\documentclass[10pt, article]{IEEEtran}

\IEEEoverridecommandlockouts
\usepackage{amsmath,graphicx}

\usepackage[dvipsnames]{xcolor}

\usepackage{amssymb,amsfonts}
\usepackage{graphicx}
\usepackage{textcomp}
\usepackage{subfigure}
\usepackage{textcomp} 
\usepackage{siunitx}  
\usepackage{algorithm}
\usepackage{algpseudocode}
\usepackage{xpatch}
\usepackage{pgfplots}
\usepackage{pgfplotstable}
\pgfplotsset{compat=1.18}

\usepackage{readarray}
\usepackage{siunitx}
\usepackage{bm}
\usepackage[nolist]{acronym}
\usepackage{bbm}
\usepackage{gensymb}
\usepackage{amsthm}

\usepackage{caption}

\usetikzlibrary{external}
\usetikzlibrary{spy}

\definecolor{TUMBeamerYellow}    {rgb} {1.000,0.706,0.000}    
\definecolor{TUMBeamerOrange}    {rgb} {1.000,0.502,0.000}    
\definecolor{TUMBeamerRed}       {rgb} {0.898,0.204,0.094}    
\definecolor{TUMBeamerDarkRed}   {rgb} {0.792,0.129,0.247}    
\definecolor{TUMBeamerBlue}      {rgb} {0.000,0.600,1.000}    
\definecolor{TUMBeamerLightBlue} {rgb} {0.255,0.745,1.000}    
\definecolor{TUMBeamerGreen}     {rgb} {0.569,0.675,0.420}    
\definecolor{TUMBeamerLightGreen}{rgb} {0.710,0.792,0.510}    

\definecolor{TUMBlue}		{cmyk}{1.00,0.43,0.00,0.00}	

\definecolor{TUMWhite}		{cmyk}{0.00,0.00,0.00,0.00}	
\definecolor{TUMBlack}		{cmyk}{0.00,0.00,0.00,1.00}	

\definecolor{TUMDarkerBlue}	{cmyk}{1.00,0.54,0.04,0.19}	
\definecolor{TUMDarkBlue}	{cmyk}{1.00,0.57,0.12,0.70}	

\definecolor{TUMDarkGray}	{cmyk}{0.00,0.00,0.00,0.80}	
\definecolor{TUMMediumGray}	{cmyk}{0.00,0.00,0.00,0.50}	
\definecolor{TUMLightGray}	{cmyk}{0.00,0.00,0.00,0.20}	

\definecolor{TUMIvony}		{cmyk}{0.03,0.04,0.14,0.08}	
\definecolor{TUMOrange}		{cmyk}{0.00,0.65,0.95,0.00}	
\definecolor{TUMGreen}		{cmyk}{0.35,0.00,1.00,0.20}	
\definecolor{TUMLightBlue}	{cmyk}{0.42,0.09,0.00,0.00}	
\definecolor{TUMLighterBlue}	{cmyk}{0.65,0.19,0.01,0.04}	

\definecolor{TUMPurple}		{cmyk}{0.50,1.00,0.00,0.40}
\definecolor{TUMDarkPurple}	{cmyk}{1.00,1.00,0.00,0.40}
\definecolor{TUMTurquois}	{cmyk}{1.00,0.03,0.30,0.30}
\definecolor{TUMDarkGreen}	{cmyk}{1.00,0.00,1.00,0.20}
\definecolor{TUMDarkerGreen}	{cmyk}{0.60,0.00,1.00,0.20}
\definecolor{TUMYellow}		{cmyk}{0.00,0.10,1.00,0.00}
\definecolor{TUMDarkYellow}	{cmyk}{0.00,0.30,1.00,0.00}
\definecolor{TUMLightRed}	{cmyk}{0.00,0.80,1.00,0.10}
\definecolor{TUMRed}		{cmyk}{0.10,1.00,1.00,0.10}
\definecolor{TUMDarkRed}	{cmyk}{0.00,1.00,1.00,0.40}
\DeclareMathOperator*{\argmin}{arg\,min}
\DeclareMathOperator{\tr}{tr}
\DeclareMathOperator{\T}{T}
\DeclareMathOperator{\He}{H}

\DeclareMathOperator{\inv}{-1}

\DeclareMathOperator{\diag}{diag}

\newcommand{\unitriag}{\mathbb{U}}

\newcommand{\eye}{\bm{\mathrm{I}}}

\newcommand{\bB}{\bm{B}}
\newcommand{\bR}{\bm{R}}

\newcommand{\Her}{{\He}}

\newcommand{\bQ}{\bm{Q}}

\newcommand{\bH}{\bm{H}}

\newcommand{\bh}{\bm{h}}

\newcommand{\cmplx}[1]{\mathbb{C}^{#1}}
\newcommand{\norm}[1]{\|#1\|}
\newcommand{\abs}[1]{|#1|}

\newcommand{\expct}[1]{\mathbb{E}[#1]}

\newcommand{\gaussdist}[2]{\mathcal{N}_{\mathbb{C}}(#1,#2)}
\newcommand{\summe}[2]{\sum_{#1}^{#2}}

\newcommand{\intround}[1]{\left\lceil #1 \right \rfloor}
\newcommand{\Mod}[1]{\mathrm{Mod}\left(#1 \right)}

\newtheorem{theorem}{Theorem}
\newtheorem{lemma}{Lemma}
\newtheorem{proposition}{Proposition}
\newtheorem{corollary}{Corollary}
\newtheorem{remark}{Remark}
\newtheorem{definition}{Definition}

\begin{document}

\begin{acronym}
    \acro{AoD}{angle of departure}
    \acro{AoA}{angle of arrival}
    \acro{ULA}{uniform linear array}
    \acro{CSI}{channel state information}
    \acro{LOS}{line of sight}
    \acro{EVD}{eigenvalue decomposition}
    \acro{BS}{base station}
    \acro{MS}{mobile station}
    \acro{mmWave}{millimeter wave}
    \acro{DPC}{dirty paper coding}
    \acro{IRS}{intelligent reflecting surface}
    \acro{AWGN}{additive white gaussian noise}
    \acro{MIMO}{multiple-input multiple-output}
    \acro{UL}{uplink}
    \acro{DL}{downlink}
    \acro{OFDM}{orthogonal frequency-division multiplexing}
    \acro{TDD}{time-division duplex}
    \acro{LS}{least squares}
    \acro{MMSE}{minimum mean square error}
    \acro{SINR}{signal to interference plus noise ratio}
    \acro{OBP}{optimal bilinear precoder}
    \acro{LMMSE}{linear minimum mean square error}
    \acro{MRT}{maximum ratio transmitting}
    \acro{M-OBP}{multi-cell optimal bilinear precoder}
    \acro{S-OBP}{single-cell optimal bilinear precoder}
    \acro{SNR}{signal-to-noise-ratio}
    \acro{THP}{Tomlinson-Harashima precoding}
    \acro{dTHP}{distributed THP}
    \acro{cTHP}{centralized THP}
    \acro{RIS}{reconfigurable intelligent surface}
    \acro{SE}{spectral efficiency}
    \acro{MSE}{mean squared error}
    \acro{ASD}{angular standard deviation}
    \acro{ZF-THP}{zero-forcing THP}
    \acro{ZF}{zero-forcing}
    \acro{SD}{sphere decoding}
    \acro{LR}{Lattice reduction}
    \acro{SER}{symbol error rate}
    \acro{BER}{bit error rate}
    \acro{BC}{broadcast channel}
    \acro{VP}{vector perturbation}
    \acro{RO}{rounding off}
    \acro{NP}{nearest plane}
    \acro{LLL}{Lenstra–Lenstra–Lovász}
    \acro{FSD}{Fixed-Complexity Sphere Decoder }
    \acro{HKZ}{Hermite–Korkine–Zolotarev}
    \acro{UTLR}{Unitriangular-structured Lattice Reduction}
\end{acronym}

\title{Is Lattice Reduction Necessary\\ for Vector Perturbation Precoding?}

\author{\IEEEauthorblockN{Dominik Semmler, Wolfgang Utschick, and Michael Joham}\\
\IEEEauthorblockA{\textit{School of Computation, Information and Technology,\\ Technical University of Munich, 80333 Munich, Germany} \\
email: \{dominik.semmler,utschick,joham\}@tum.de}
}

\maketitle
\begin{abstract}
\Ac{VP} precoding is an effective nonlinear precoding technique in the \ac{DL} with modulo channels, providing an approximation of \ac{DPC} which is capacity-achieving.
Especially, when combined with \ac{LR}, low-complexity algorithms achieve a very promising performance, outperforming other popular nonlinear precoding techniques like \ac{THP}.
However, these results are based on the \ac{SER} or \ac{BER}.
When shifting the focus to the mutual information as the figure of merit, we show that this is different and that the underlying lattice problem has a unique structural property. 
For lattice problems with this special structure, we show for a whole class of algorithms that \ac{LR} does not have any impact on the solution vector.
At the same time, algorithms are identified which benefit from \ac{LR}, even if this lattice structure arises.
The provided structural analysis has strong implications on the performance evaluation of \ac{VP}.
In particular, we re-evaluate popular \ac{LLL}-aided methods like the \ac{LLL}-aided \ac{NP} algorithm and show that they do not outperform conventional \ac{THP}, highlighting the effectiveness of the \ac{THP} method. 
This is in contrast to the existing results based on \ac{SER} and \ac{BER} where these methods clearly outperform \ac{THP}.
\end{abstract}


\begin{figure}[b]
    \onecolumn
    \centering
    \scriptsize{This work has been submitted to the IEEE for possible publication. Copyright may be transferred without notice, after which this version may no longer be accessible.}
    \vspace{-1.3cm}
    \twocolumn
\end{figure}

\acresetall

\section{Introduction}
{\Ac{DPC} is known to achieve the capacity of the \ac{MIMO} \ac{BC} \cite{DPCCapacity}.
However, as it is difficult to implement,
approximations like linear precoding are often adopted.
In \cite{VP}, \ac{VP} has been proposed to improve upon linear precoding by introducing an integer vector at the \ac{BS} which reduces the energy of the transmit symbol.
This vector can be optimized without introducing interference
because the receivers are equipped with modulo operators, which remove any perturbations by integers.
Optimizing this integer vector is a lattice closest vector problem, and the \ac{SD} algorithm (see, e.g., \cite{SphereDecoder,SphereDecoderFinckePohst,SphereDecoderDetection,SphereDecoderVLSI,SphereDecoderExpCompl}) gives the optimal solution.
While this has been discussed in \cite{VP}, the \ac{SD} algorithm has an exponential complexity, rendering it impractical.
Hence, polynomial-time algorithms are preferred, and in \cite{Babai}, various algorithms, like the \ac{NP} and the \ac{RO} algorithms, have been proposed.
Most importantly, the effectiveness of applying \ac{LR} prior to the application of the \ac{NP} or \ac{RO} method has been demonstrated in \cite{Babai}.
\ac{LR} (see \cite{LLL}) is an elegant way of optimizing a lattice problem, transforming the original lattice basis into one with short and nearly-orthogonal basis vectors without changing the underlying problem.
After this reduction, already very simple methods are significantly improved and achieve a good performance \cite{Babai} (see, e.g., \cite{LRDetection} for \ac{LR} in \ac{MIMO} detection).
The popular \ac{THP} method (see \cite{THPOne,THPTwo,THPRates,THPMichael,THPM}) can also be interpreted as a particular sub-optimal technique for \ac{VP}, resulting in a specific integer vector.
This optimization concept of \ac{THP} is very similar to the \ac{NP} method \cite{Babai}, and it can also be combined with \ac{LR}. 
}

{
It has already been shown (see, e.g., \cite{WubbenLLLBetter,LLLBetterTwo,LLLBetterThree}) that \ac{LR} techniques are also powerful for \ac{VP} and clearly improve upon conventional methods, including \ac{THP}, albeit allowing a polynomial complexity.
However, these results were based on error rates, whereas in this article, we consider the mutual information with an infinite constellation representing the maximum number of bits that can be transferred per channel use by the communication system for an arbitrarily small error probability.
The mutual information has the clear advantage that it provides a general measure of information for which the system can be optimized instead of relying on specific choices for channel codes (or on an uncoded system).
In \cite{VPSumRate}, \ac{VP} has already been analyzed when considering the mutual information as the metric.
\cite{VPSumRate} is primarly based on the \ac{ZF} precoder whereas the analysis was extended to \ac{MMSE} approaches in \cite{VPMMSE} (based on MMSE VP which was introduced in \cite{MMSEVPDavid}).
Furthermore, in \cite{VPRateMultiDimLattice}, \ac{VP} was generalized to a processing over multiple time instances by also considering the mutual information with \ac{ZF} precoders.
Integer-Forcing is another related concept using lattice optimization which has been analyzed in, e.g., \cite{IntegerForcing,IntegerForcingULDL,IntegerForcingTWC,IntegerForcingRx}.
For the mutual information, the \ac{LR}-aided \ac{RO} method with dithering (which is not assumed in this article) has been shown to be inferior to \ac{THP} in \cite{DitheringRO}.
Hence, it is not clear if \ac{LR} still leads to significant gains under that metric.
This article includes a re-evaluation of \ac{LR}-based methods for the mutual information.
}

{
We focus on classical \ac{VP} precoding according to \cite{VP} where we also use the \ac{ZF} precoding matrix given by the pseudoinverse of the channel.
However, an optimized diagonal matrix is incorporated in front of the pseudoinverse.
This optimization of the diagonal matrix is often neglected in the literature.
In \cite{VPSumRate}, this matrix has been discussed by using a lower bound on the symbol energy (see \cite{EnergyBound}).
However, according to this bound, the optimization of the diagonal matrix is not important.
We will show that when considering actual algorithms that the choice of this diagonal matrix is crucial, especially for ill-conditioned channels.
}

{
Beyond the improved performance, the optimized allocation matrix results in the interesting aspect that the arising lattice problem shows a special structural property.
Under this structure, we will show that \ac{LR} has no impact on a whole class of algorithms (including \ac{NP}).
The effectiveness of \ac{LLL} reduction has already been analyzed in \cite{LLLAnalysis} for the \ac{NP} algorithm.
However, a specific setup of a detection scenario with Gaussian noise has been assumed where the success probability of the \ac{NP} algorithm was the figure of merit.
With the power minimization in the \ac{DL}, we have a different setup with the success probability not existing in our case. 
Therefore, this analysis is not applicable.
}

{
Our considerations are primarily based on a specific mathematical property of the lattice basis.
Apart from this structural constraint, a general lattice problem is considered.
This allows us to generalize our results to a complete class of optimization algorithms together with different \ac{LR} algorithms.
As no specific model is assumed, our analysis is valid for any lattice problem where the lattice basis has this structural property.
We consider more advanced algorithms beyond the \ac{RO} and \ac{NP} algorithms.
Specifically, apart from  \ac{SD}, we focus on two very powerful and popular algorithms
These are the K-Best (see, e.g., \cite{KBestVLSI,KBestImpl,KBestImp,FPGAQRDMFSD,KBestLLLOne,KBestLLLTwo,KBestLLLThree}) and \ac{FSD} (see, e.g., \cite{FSD,FSDTwo,FSDThree,FSDFour,FSE,FSETwo,FSEThree}) algorithms.
}

{In summary, this paper has the following contributions:
\begin{itemize}
    \item We identify a special lattice structure naturally arising from an optimized \ac{ZF} precoder in \ac{VP} precoding.
    Under this structure, we show that \ac{LLL} reduction yields no improvement for \ac{NP} and \ac{THP}.
    This is in contrast to the existing literature, where it is assumed that, especially, the \ac{LLL}-aided \ac{NP} (\ac{THP}) method clearly improve over conventional \ac{THP}.
    Hence, this structural discovery has strong implications on the performance of (\ac{LR}-aided) \ac{VP}.
       \item The results above are generalized to a whole subset of lattice problems characterized by a structural constraint on the lattice basis.
    We show for a complete class of algorithms ({UTLR-invariant} schemes) that \ac{LR} yields no improvement (this includes the \ac{NP}, \ac{SD}, and K-Best algorithm) for this subset of lattice problems.
    At the same time, this allows to identify algorithms which can still benefit from \ac{LR} (\ac{RO}, \ac{FSD}) despite the structural characteristic of the lattice basis. 
    Moreover, we further generalize the results to \ac{HKZ} lattice reduction which is considerably more complex than \ac{LLL}.
    Additionally, the results are generalized to the whole rate region as well as to more practical considerations, highlighting the applicability beyond theoretical assumptions.
\end{itemize}}
\begin{table}[t!]
    \centering
            \begin{tabular}{|l|l|l|}
    \hline
    $\bm{s}$ && information symbols (sent by BS)\\
    \hline
    $\bm{\tilde{s}}$&&received information symbols (at the users)\\
    \hline
   $\bm{\bar{x}}$&$ \bm{P}(\bm{s} + \bm{a})$& precoded transmit symbols\\
    \hline
    $\bm{{x}}$&$\rho \bm{\bar{x}}$& normalized $\bm{\bar{x}}$ (power constraint with equality)\\
    \hline
\end{tabular}
\caption{Notation Table.}
\label{Table:NotationTable}
\end{table}
\color{black}
\section{System Model}
We consider $K$ single-antenna users being served by one \ac{BS} having $N$ antennas.
The channel of the $k$-th user is denoted as $\bh_{\mathbb{C},k}^{\Her} \in \cmplx{1\times N}$, $y_{\mathbb{C},k}$ is the complex received signal of user $k$, and $\bm{x}_{\mathbb{C}}$ is the complex transmit symbol at the base station.
This results in the system model
\begin{equation}
    \bm{y}_{\mathbb{C}} = \bm{H}_{\mathbb{C}} \bm{x}_{\mathbb{C}} + \bm{n}_{\mathbb{C}} \in \mathbb{C}^K
\end{equation}
where $\bm{n}_{\mathbb{C}} \sim \gaussdist{\bm{0}}{\eye}$ is the complex Gaussian noise vector.
By stacking the real and imaginary parts as $\bm{{y}} = [\mathrm{Re}(\bm{y}^{\T}_{\mathbb{C}}),\mathrm{Im}(\bm{y}^{\T}_{\mathbb{C}})]^{\T}$, $\bm{x} = [\mathrm{Re}(\bm{x}^{\T}_{\mathbb{C}}),\mathrm{Im}(\bm{x}^{\T}_{\mathbb{C}})]^{\T}$, and $\bm{n} = [\mathrm{Re}(\bm{n}^{\T}_{\mathbb{C}}),\mathrm{Im}(\bm{n}^{\T}_{\mathbb{C}})]^{\T} \sim \mathcal{N}(\bm{0},\eye \frac{1}{2})$, 
we arrive at the equivalent real-valued representation
\begin{equation}
    \bm{{y}} = \bm{H}\bm{x} + \bm{n} \in \mathbb{R}^{2K}, \quad \text{with} \quad \bm{H} = \left[\begin{smallmatrix}
        \mathrm{Re}(\bm{H}_{\mathbb{C}})& -\mathrm{Im}(\bm{H}_{\mathbb{C}})\\
        \mathrm{Im}(\bm{H}_{\mathbb{C}})&\mathrm{Re}(\bm{H}_{\mathbb{C}})
    \end{smallmatrix}\right].
\end{equation}
Considering a modulo channel, each receiver is equipped with a modulo operator located after the receive filter.
Additionally, the transmit symbol is constructed as
\begin{equation}
    \bm{x} = \rho \bm{\bar{x}} =  \rho \bm{P}(\bm{s} + \bm{a})\in \mathbb{R}^{2 N}
\end{equation}
where $\bm{P}\in \mathbb{R}^{2 N \times 2K}$ is the precoder and $\bm{s} \in \mathbb{R}^{2K}$ is the symbol vector where $s_k$ and $s_{k+K}$ is respectively the real and imaginary part of the symbol for user $k$.
The characteristic of the modulo channel is the integer vector $\bm{a} \in \mathbb{Z}^{2K}$, which is added to the symbol vector and is designed to optimize the performance.
Furthermore, the scaling $\rho = \sqrt{{P_{\text{Tx}}}/{\expct{\norm{\bm{\bar{x}}}^2}}}$ is introduced such that the transmit power constraint $ \expct{\norm{\bm{x}}^2} = \expct{\norm{\bm{\bar{x}}}^2} \rho^2 \le P_{\text{Tx}}$ is fulfilled with equality.
Combining these definitions, the system model reads as 
\begin{equation}
    \bm{\tilde{s}} = \Mod{ \bm{G} \bm{H} \bm{P} ( \bm{s} + \bm{a}) \rho + \bm{G} \bm{n}}
\end{equation}
where $\bm{G}$ is the diagonal receive filter and $ \bm{\tilde{s}}$ is the vector after the modulo operators{(please see TABLE \ref{Table:NotationTable} for an overview of the variables)}.
{\ac{VP} is especially suited for interference cancellation in the high \ac{SNR} regime.
Hence, we choose the \ac{ZF} solution (as in \cite{VP,VPSumRate,VPSumRateHighSNR})}
\begin{equation}
    \bm{P} = \bm{H}^+ \bm{D}
\end{equation}
where $\bm{H}^+$ is the pseudoinverse of the channel and $\bm{D}= \diag(d_1,\dots,d_{2K})$ is the diagonal rate allocation matrix
which we show to be an essential aspect for the performance.
Considering a general matrix $\bm{D}$, the receive filter {reads}
\begin{equation}
    \bm{G} = \bm{D}^{\inv} \frac{1}{\rho}
\end{equation}
after which the received signal can be expressed as 
\begin{equation}
    \begin{aligned}
        \bm{\tilde{s}}&= \Mod{   \bm{s} + {\bm{D}^{\inv} \bm{n}}/{\rho} }
    \end{aligned}
\end{equation}
where the integer vector $\bm{a}$ vanishes due to the modulo operations.
{Throughout the article, we use uniformly distributed symbols $s_i$ between $-0.5$ and $0.5$ (see, e.g., \cite{Forney,VPRateMultiDimLattice,VPSumRate,VPSumRateHighSNR}) denoted as $s_i \sim \mathcal{U}(-0.5,0.5)$. 
However, our main results can be generalized to arbitrary input distributions and practical input distributions will be discussed in this article}.
Following \cite{VPSumRate}, the mutual information reads as
\begin{equation}\label{eq:RateNonHSNR}
    \begin{aligned}
        I(\bm{\tilde{s}};\bm{s})
        &= - \summe{k=1}{2K} h\Big( \mathrm{Mod}\big(   s_k +  {n_k}/(\rho d_k)\big)\Big).
    \end{aligned}
\end{equation}
In this article, our main focus is on the high-\ac{SNR} regime where the asymptote of $I(\bm{\tilde{s}};\bm{s})$ (see \cite[Section III]{VPSumRate}) can be expressed as
{\begin{equation}\label{eq:RateHSNR}
    \begin{aligned}
    R^{\text{HSNR}} 
    &=  \underset{\bm{a},\bm{D}}{\max} \quad K\log_2 \left( \frac{P_{\text{Tx}}}{\pi e \expct{\norm{\bm{\bar{x}}}^2} }\sqrt[2K]{\prod\nolimits_{k=1}^{2K}d_k^2}\right).
    \end{aligned}
\end{equation}}
So far, we have been following \cite{VPSumRate,VPSumRateHighSNR} for the system model.
This high-\ac{SNR} asymptote \eqref{eq:RateHSNR} is the utility of this article and jointly optimizing this expression w.r.t. the scaling $\bm{D}$ and the integer values $\bm{a}$ is analyzed in the following.
{When optimizing this expression w.r.t. $\bm{D}$, it is important to note that $\bm{\bar{x}}=\bm{H}^+ \bm{D} (\bm{s} + \bm{a})$ depends on $\bm{D}$ which makes the optimization non-trivial.
This joint optimization is typically omitted in the literature, and in, e.g., \cite{VP}, $\bm{D}$ is chosen to have fixed values, particularly, $\bm{D} = \eye$.
On the other hand, \cite{VPSumRate} addressed the joint optimization of $\bm{D}$ and $\bm{a}$ by using the bound of \cite[Lemma 1]{EnergyBound}}
\begin{equation}\label{eq:SymbolEnergyBound}
    \begin{aligned}
        \expct{\norm{\bm{\bar{x}}}^2} &\ge \frac{K \Gamma(K+1)^{\frac{1}{K}}}{(K+1)\pi} \det(\bm{D}(\bm{HH}^{\T})^{\inv} \bm{D})^{\frac{1}{2K}}.\\
    \end{aligned}
\end{equation}
{
    It has to be noted that this bound only holds for signals being uniformly distributed over the Voronoi cell and, hence, only serves as an approximation.
}
Substituting this expression into \eqref{eq:RateHSNR} results in the upper bound
\begin{equation}
    \begin{aligned}\label{eq:RateUB}
        R^{\text{HSNR}}_{\text{UB}}
        & = \log_2\det(\bm{H}_{\mathbb{C}}\bm{H}_{\mathbb{C}}^{\Her} \bar{p} ) -  K\log_2 \left(  \frac{ e \Gamma(K+1)^{\frac{1}{K}}} {K+1} \right) \\
\end{aligned}
\end{equation}
with $\bar{p} = P_{\text{Tx}}/K$. 
Hence, following \cite{VPSumRate}, it can be observed that optimizing the scaling matrix $\bm{D}$ is not important at high-\ac{SNR} as \eqref{eq:RateUB} does not depend on $\bm{D}$.
However, the importance of this result highly depends on the tightness of the bound in \eqref{eq:SymbolEnergyBound}.
We will show in the following that when using actual algorithms, including the optimal \ac{SD} algorithm, the choice of $\bm{D}$ does matter.
\section{Conventional Algorithms}
A variety of different algorithms can be used to minimize
\begin{equation}\label{eq:symbol_energy}
    \expct{\norm{\bm{\bar{x}}}^2} = \expct{\norm{\bm{H}^+ \bm{D} (\bm{s} + \bm{a})}^2}
\end{equation}
w.r.t. the integer values $\bm{a}$.
The optimal choice can be found by \ac{SD} and has been analyzed in \cite{VP} with the identity scaling $\bm{D} = \eye$.
Due to its exponential complexity, we focus on the algorithms in \cite{Babai}, 
where at first we consider the conventional methods without \ac{LR}.
\subsection{Rounding Off}
The \ac{RO} procedure (see \cite{Babai}) is probably the simplest method.
Here, each $a_i$ is chosen as the nearest integer, i.e., 
\begin{equation}\label{eq:ROZeros}
    a_i^{\textnormal{RO}} = \intround{s_i} = 0
\end{equation}
where $\intround{\bullet}$ is the rounding operation to the nearest integer.
Therefore, the solution results in the zero vector as the symbols $s_i$ are assumed to be distributed between $-0.5$ and $0.5$.
As $\bm{a} = \bm{0}$, the solution results in classic linear zero-forcing beamforming.
\subsection{Nearest Plane}
The \ac{NP} algorithm of \cite{Babai} is a considerable improvement.
With the QR decomposition $\bm{H}^+ \bm{D} =  \bm{{Q}}\bm{{R}} \bm{D} = \bm{{Q}}\bm{\tilde{R}}$, the integers $a_i$ are now successively chosen as
\begin{equation}\label{eq:SolutionNP}
    {a}^{\textnormal{NP}}_i = -\intround{{s}_i +\summe{j>i}{}\frac{\tilde{r}_{ij}}{\tilde{r}_{ii}} ({s}_j + {a}_j)}
\end{equation}
by starting with $i=2K$ until $i=1$.
\subsection{THP}
Finally, also the popular \ac{THP} algorithm can be interpreted as a method which optimizes \eqref{eq:symbol_energy} w.r.t. the integer vector $\bm{a}$.
Interestingly, it follows the same procedure as the \ac{NP} algorithm.
The only difference is that \ac{THP} works on a lower triangular matrix, whereas the \ac{NP} algorithm uses an upper triangular matrix.
Hence, for \ac{THP}, we use an LQ decomposition $\bm{H} = \bm{LQ}$ and with $\bm{H}^+ \bm{D} = \bm{Q}^{\T} \bm{L}^{\inv} \bm{D} =  \bm{Q}^{\T}\bm{\tilde{L}}^{\inv}$, the integers $a_i$ are successively updated as
\begin{equation}\label{eq:THPSolution}
    {a}^{\textnormal{THP}}_i = -\intround{{s}_i +\summe{j<i}{}\frac{\tilde{l}_{ij}^{\inv}}{\tilde{l}_{ii}^{\inv}} ({s}_j + {a}_j)}
\end{equation}
by starting with $i=1$ until $i=2K$ where $\tilde{l}_{ij}^{\inv}$ is the $ij$-th entry of $\bm{\tilde{L}}^{\inv}$. 
\subsection{Overview}
It is important to note that we generalized the algorithms above w.r.t. an arbitrary scaling matrix $\bm{D}$.
Very often, the algorithms are given with the identity scaling $\bm{D} = \eye$.
Besides, the THP algorithm is also given with the scaling $\bm{D} = \diag(\bm{l})$ where $\bm{l}$ is the diagonal of $\bm{L}$ (see, e.g., \cite{MultiBranchTHP} for both scalings).
We will use the generalized versions with an arbitrary $\bm{D}$, and we will derive the optimal solution of $\bm{D}$ for each of the algorithms in the following.
As the THP and the \ac{NP} algorithms follow the same concept and have the same performance, we give all proofs only once for the \ac{NP} algorithm.
However, all results can be immediately extended to THP.

\section{Performance of Conventional Algorithms}
In this section, we derive the optimal scaling matrix $\bm{D}^{{\star}}$ and the corresponding rates for the different algorithms in closed form.
For the derivations, we need the following lemma.
\begin{lemma}\label{lemma:ModuloStayUniform}
    For two independent random variables $X$ and $N$ where $X ~\sim~\mathcal{U}(-0.5,0.5)$, the distribution of $Y=\textnormal{Mod}(X+N)$ is again uniform between $-0.5$ and $0.5$.
\end{lemma}
\begin{proof}
    This is a well-known result. We gave a short proof in \cite{NonlinearRIS}.
\end{proof}
\subsection{Optimized Rounding Off}
The \ac{RO} algorithm coincides with the linear zero-forcing beamforming solution.
Even though we are assuming a uniform distribution instead of a Gaussian distribution, this does not affect the optimal scaling $\bm{D}^{{\star}}$.
Also, the rate resembles that of linear zero-forcing at high-\ac{SNR} with the difference that we have a shaping loss of $K\log_2\left(\frac{\pi e }{6 } \right)$ bits as the uniform distribution of the data symbols is not optimal.
This observation is summarized in the following proposition.
\begin{proposition}
    The optimal scaling matrix $\bm{D}^{{\star}}$ when using the conventional \ac{RO} algorithm is given by 
    \begin{equation}\label{eq:RoundingOffOptScaling}
        \bm{D}^{{\star}} = \diag((\bm{H}\bm{H}^{\T})^{\inv})^{\frac{1}{2}}.
    \end{equation}
    For this optimal choice of the scaling matrix, the \ac{RO} rate reads as
    \begin{equation}\label{eq:RoundingOffOptRate}
        R_{\textnormal{RO}}^{\textnormal{HSNR}} =  \summe{i=1}{K}\log_2\left(\frac{\bar{p}}{\bm{e}_i^{\T}(\bm{H}_{\mathbb{C}} \bm{H}_{\mathbb{C}}^{\Her})^{\inv}\bm{e}_i}\right)-K\log_2\left(\frac{\pi e }{6 } \right).
    \end{equation}
\end{proposition}
\begin{proof}
For the \ac{RO} solution ($\bm{a} = \bm{0}$), the transmit signal reduces to 
\begin{equation}
    \bm{x} = \bm{H}^+ \bm{D} (\bm{s} + \bm{a}) =  \bm{H}^+ \bm{D} \bm{s}.
\end{equation}
It follows that the signal power is given by 
\begin{equation}
    \expct{\norm{\bm{x}}^2} = \tr( \bm{H}^+ \bm{D} \expct{\bm{s}  \bm{s}^{\T}}\bm{D} \bm{H}^{+,\T}) = \tr(\bm{D}(\bm{H} \bm{H}^{\T})^{\inv}\bm{D})/12
\end{equation}
as the symbols $s_i$ are independent of each other and have variance $\expct{s_i^2} = \frac{1}{12}$.
Using the arithmetic-geometric mean inequality yields
    \begin{equation}
        \begin{aligned}
            \tr(\bm{D}(\bm{H} \bm{H}^{\T})^{\inv}\bm{D}) &= \summe{i=1}{2K}d_i^2 \bm{e}_i^{\T} (\bm{H} \bm{H}^{\T})^{\inv} \bm{e}_i \\
            &\ge 2K \sqrt[2K]{\prod \nolimits_{i=1}^{2K}d_i^2 \bm{e}_i^{\T} (\bm{H} \bm{H}^{\T})^{\inv} \bm{e}_i}
        \end{aligned}
    \end{equation}
    {and} we can upper bound the rate [cf. \eqref{eq:RateHSNR}] as 
    \begin{equation}
    \begin{aligned}
        R^{\text{HSNR}}_{\text{RO}}&\le K\log_2\left(\frac{ 12 P_{\text{Tx}} }{2 K\pi e} \sqrt[2K]{\prod \nolimits_{i=1}^{2K} \left(\bm{e}_i^{\T} (\bm{H} \bm{H}^{\T})^{\inv} \bm{e}_i\right)^{\inv}}\right)\\
    \end{aligned}
\end{equation}
which can be rewritten as in \eqref{eq:RoundingOffOptRate}.
This upper bound is achieved for 
   $ d_i = \pm {1}/{\sqrt{\bm{e}_i^{\T} (\bm{H} \bm{H}^{\T})^{\inv} \bm{e}_i}}$.
Hence, the scaling in \eqref{eq:RoundingOffOptScaling} is an optimal solution.
Note that the sign of $d_i$ can be arbitrarily chosen and does not affect the rate.
\end{proof}

\subsection{Optimized Nearest Plane}
The \ac{NP} algorithm improves on the simple \ac{RO} procedure.
Interestingly, in comparison to \ac{RO}, which has the same performance as linear precoding with an additional shaping loss, 
the \ac{NP} algorithm has the same performance as \ac{DPC} with the same shaping loss.
This is shown in the following, where we also give the optimal matrix $\bm{D}^{{\star}}$.
\begin{proposition}\label{prop:ScalingMatrixOptRate}
    The optimal scaling matrix $\bm{D}^{{\star}}$, when using the conventional \ac{NP} Algorithm, is given by 
    \begin{equation}\label{eq:NearestPlaneOptScaling}
        \bm{D}^{{\star}} = \diag^{\inv}(\bm{r})
    \end{equation}
    where $\bm{r}$ is the diagonal of $\bm{R}$ from the QR decomposition $\bm{H}^+ = \bm{QR}$.
    For this optimal choice of the scaling matrix, the rate reads as
    \begin{equation}\label{eq:NearestPlaneOptRate}
        R_{\textnormal{NP}}^{\textnormal{HSNR}} =      \log_2\det\left(\bar{p} \bH_{\mathbb{C}} \bH_{\mathbb{C}}^\Her \right) -   K\log_2\left(\frac{\pi e}{6}\right).
    \end{equation}
\end{proposition}
\begin{proof}
    With the QR decomposition $\bm{H}^+  = \bm{QR}$, we arrive at
    \begin{equation}
        \begin{aligned}
            \bm{x} &= \bm{H}^+\bm{D}(\bm{s} + \bm{a})=\bm{QR}\bm{D}(\bm{{s}} + \bm{{a}}).
        \end{aligned}
    \end{equation}
    According to the \ac{NP} algorithm in \eqref{eq:SolutionNP}, we find
    \begin{equation}\label{eq:symbolenergyNP}
        \begin{aligned}
            \bm{e}_i^{\T} \bm{{R}D}(\bm{{s}} + \bm{{a}}) &={r}_{ii}d_i\left( \summe{j>i}{}\frac{{r}_{ij}d_j}{{r}_{ii}d_i} ({s}_j + {a}_j) + {s}_i +{a}_i\right)\\
            &={r}_{ii}d_i \mathrm{Mod}\left(w_i + {s}_i\right)\\
        \end{aligned}
    \end{equation}
    where the last line follows from \eqref{eq:SolutionNP} with $w_i =  \summe{j>i}{}\frac{{r}_{ij}d_j}{{r}_{ii}d_i} ({s}_j + {a}_j)$.
    According to Lemma \ref{lemma:ModuloStayUniform}, as $s_{i}$ is independent of $w_i$, also $\mathrm{Mod}\left(w_i + {s}_i\right) \sim \mathcal{U}(-0.5,0.5)$. 
    It follows that $\expct{\abs{ \bm{e}_i^{\T} \bm{{{R}}D}(\bm{{s}} + \bm{{a}})}^2} = r^2_{ii}d^2_i \frac{1}{12}$ and, therefore, we get 
    \begin{equation}
        \expct{\norm{ \bm{Q{R}D}(\bm{s} + \bm{a})}^2} = \expct{\norm{ \bm{{R}D}(\bm{s} + \bm{a})}^2} =  \frac{1}{12} \summe{i=1}{2K} {r}_{ii}^2 d_i^2.
    \end{equation}
    Applying the arithmetic-geometric mean inequality
    \begin{equation}
        \begin{aligned}
            \summe{i=1}{2K} {r}_{ii}^2 d_i^2 \ge 2K \sqrt[2K]{ \prod\nolimits_{i=1}^{2K} {r}_{ii}^2 d_i^2 }
        \end{aligned}
    \end{equation}
    we can upper bound the rate as [cf. \eqref{eq:RateHSNR}] 
    \begin{equation}
    \begin{aligned}
        R^{\text{HSNR}}_{\text{NP}}&\le K\log_2\left(\frac{ 12 P_{\text{Tx}} }{2 K\pi e} \sqrt[2K]{\det(\bm{H}\bm{H}^{\T})}\right)\\
    \end{aligned}
\end{equation}
which can be rewritten as in \eqref{eq:NearestPlaneOptRate}.
This upper bound is achieved for 
    $d_i = \pm {1}/{r_{ii}}$
and, hence, the scaling in \eqref{eq:NearestPlaneOptScaling} is optimal.
\end{proof} 
\begin{corollary}
    The same Rate can be achieved by \ac{THP},
    \begin{equation}
        R_{\textnormal{THP}}^{\textnormal{HSNR}} = R_{\textnormal{NP}}^{\textnormal{HSNR}}.
    \end{equation}
    THP follows the same successive procedure as the \ac{NP} algorithm, with the only difference that a lower instead of an upper triangular matrix is used.
    For THP, the optimal scaling matrix is, therefore,
    \begin{equation}\label{eq:THPOptScaling}
        \bm{D}^{{\star}} = \diag(\bm{l})
    \end{equation}
    where $\bm{l}$ is the diagonal of the LQ decomposition $\bm{H} = \bm{L} \bm{Q}$ (which means that $\bm{H}^+= \bm{Q}^{\T}\bm{L}^{\inv} $).
\end{corollary}

\section{Performance of LLL-Aided Algorithms}
It has to be noted that the \ac{RO} and the \ac{NP} algorithms are typically combined with \ac{LR} (see \cite{Babai}) by using the \ac{LLL} algorithm \cite{LLL}.
This clearly improves the performance of the algorithms and has been demonstrated w.r.t the \ac{SER} and \ac{BER} in \cite{WubbenLLLBetter,LLLBetterTwo,LLLBetterThree}.
In this case, the solution of the \ac{RO} method is not the zero-vector anymore and, hence, the LLL-aided \ac{RO} improves upon the linear zero-forcing method.
However, while the \ac{RO} procedure is improved, we will see that the solution of the \ac{NP} (or \ac{THP}) cannot be improved by simply using \ac{LR}.
To this end, we need the following lemma. 
\begin{lemma}\label{lem:LLL_Ordered}
    Let $\bB$ be the basis of a lattice $\mathcal{L}$ with QR decomposition $\bB = \bQ \bR$ and 
    and let $\bm{\tilde{B}}  = \bB \bm{T} = \bm{\tilde{Q}} \bm{\tilde{R}}$ be the QR decomposition of the LLL-reduced basis $\bm{\tilde{B}}$, where $\bm{T}$ is a unimodular integer matrix.
    Then, if $\abs{r_{ii}} \le \abs{r_{jj}}, \, \text{for} \, i<j$, it follows that $\tilde{r}_{ii} = r_{ii}\, \forall i$ and that $\bm{T}$ is an upper triangular matrix with $t_{ii} = 1, \, \forall i$.
\end{lemma}
\begin{proof}
    We follow the LLL algorithm of \cite{LLL} with the code given in \cite{LLLPseudo} where $k$ is iterated from $k=2$ until $k=n$.
    The algorithm is initialized with $k=2$, $\bm{T} = \eye$, $\bm{\tilde{Q}} = \bQ$, and $\bm{\tilde{R}} = \bR$.
    According to the LLL-algorithm, the $k$-th columns of $\bm{\tilde{R}}$ and $\bm{T}$ are updated as
    \begin{align}\label{eq:LLLFirstStep}
        \bm{\tilde{r}}_{1:l,k} = \bm{\tilde{r}}_{1:l,k} - \intround{ \frac{\tilde{r}_{lk}}{\tilde{r}_{ll}} }\bm{\tilde{r}}_{1:l,l} \quad \text{and} \quad 
        \bm{t}_k = \bm{t}_k - \intround{ \frac{\tilde{r}_{lk}}{\tilde{r}_{ll}} } \bm{t}_l
    \end{align}
    where the variable $l$ is decreased from $k-1$ to $1$ (see \cite{LLLPseudo}).
    Combining \eqref{eq:LLLFirstStep} with the fact that $l\le k-1$, we can see that only the first $k-1$ entries of $\bm{\tilde{r}}_k$ and $\bm{t}_k$ are modified.
    Hence, the upper triangular structure of $\bm{\tilde{R}}$ and $\bm{T}$ is preserved with this step.
    Additionally,
        \begin{equation}\label{eq:LemLLLOrderdCheckCond}
        \tilde{r}^2_{k-1,k-1} \le   \tilde{r}^2_{kk} +  \tilde{r}^2_{k-1,k}
    \end{equation}
    holds as $\tilde{r}^2_{k-1,k-1} \le  \tilde{r}^2_{kk} $ by assumption.
    Therefore, the current iteration is completed \cite{LLLPseudo} and $k$ is incremented to $k+1$.
    In addition to the upper triangular structure, the diagonal of $\bm{T}$ and $\bm{\tilde{R}}$ was preserved.
    These conservations hold for all iterations and, hence, the final matrices $\bm{\tilde{R}}$ and $\bm{T}$ are upper triangular with $t_{ii} = 1$, $\tilde{r}_{ii} = r_{ii},\, \forall i$.
\end{proof}
\begin{remark}
    Please note that if the scaling in \eqref{eq:NearestPlaneOptScaling} is used, the matrix $\bm{R} \diag^{\inv}(\bm{r})$ has only ones on the diagonal and, hence, the condition in Lemma \ref{lem:LLL_Ordered} is fulfilled.
\end{remark}
{
   Furthermore, in order to proof optimality later on, we need, additionally, the following lemma.
   \begin{lemma}\label{lem:ConstantLRMatrix}
        Given a point $\bm{\bar{d}}$ such that $\bm{Q}\bm{R}\diag(\bm{\bar{d}})$ with $\abs{r_{ii}\bar{d}_i} \le \abs{r_{jj} \bar{d}_j}$ for $i<j$ ($\bm{\bar{d}}$ fulfills the condition of Lemma \ref{lem:LLL_Ordered}).
        Then, the unimodular integer matrix $\bm{T}$ of the \ac{LLL} algorithm is locally constant, i.e. 
        \begin{equation}
            \begin{aligned}
                           \exists \delta > 0:\quad &\bm{T}\left(\bm{\bar{d}} +  \bm{\Delta} \right) =  \bm{T}\left(\bm{\bar{d}} \right)\\
            & \forall \bm{\Delta}, \,\norm{\bm{\Delta}}_2 \le \delta
            \end{aligned}
        \end{equation}
        \begin{proof}
            The triangular matrix is initialized as
            \begin{equation}\label{eq:LemmaLLLConstTriangInit}
                \bm{\tilde{R}}(\bm{\Delta}) = \bm{R}\diag(\bm{\bar{d}} +  \bm{\Delta} )
            \end{equation}
            where $\bm{\tilde{R}}(\bm{0})$ is the original matrix at the point $\bm{\bar{d}}$.
            We show now that throughout the algorithm $\bm{\tilde{R}}(\bm{\Delta})$ stays an upper triangular matrix with $\tilde{r}_{ii} = r_{ii}(d_i + \Delta_i)$ and additionally fulfills
            \begin{equation}\label{eq:LemmaLLLConstAsymR}
                \underset{\norm{\bm{\Delta}}_2 \rightarrow 0}{\lim}\bm{\tilde{R}}(\bm{\Delta}) = \bm{\tilde{R}}(\bm{0}).
            \end{equation}
            For the initial matrix in \eqref{eq:LemmaLLLConstTriangInit}, \eqref{eq:LemmaLLLConstAsymR} is fulfilled.
            Considering \eqref{eq:LLLFirstStep} of the LLL algorithm in Lemma \ref{lem:LLL_Ordered}, we can always find a small enough $\norm{\bm{\Delta}}_2$ such that
            \begin{equation}\label{eq:LemmaLLLConstSameIntMatrix}
                \intround{ \frac{\tilde{r}_{lk}(\bm{\Delta})}{\tilde{r}_{ll}(\bm{\Delta})} } = \intround{ \frac{\tilde{r}_{lk}(\bm{0})}{\tilde{r}_{ll}(\bm{0})} }
            \end{equation}
            due to \eqref{eq:LemmaLLLConstAsymR} and the rounding operation.
            With \eqref{eq:LemmaLLLConstSameIntMatrix}, the matrix is modified as 
            \begin{equation}
                \bm{\tilde{r}}_{1:l,k}(\bm{\Delta}) = \bm{\tilde{r}}_{1:l,k}(\bm{\Delta}) - \intround{ \frac{\tilde{r}_{lk}(\bm{0})}{\tilde{r}_{ll}(\bm{0})} }\bm{\tilde{r}}_{1:l,l}(\bm{\Delta})
            \end{equation}
            and after this modification the new matrix $\bm{\tilde{R}}(\bm{\Delta})$ still fulfills \eqref{eq:LemmaLLLConstAsymR}.
            Additionally, it follows from \eqref{eq:LemmaLLLConstSameIntMatrix} that $\bm{T}$ is modified as
            \begin{equation}
                \bm{t}_k = \bm{t}_k - \intround{ \frac{\tilde{r}_{lk}(\bm{0})}{\tilde{r}_{ll}(\bm{0})} } \bm{t}_l
            \end{equation}
            and, therefore, the integer matrix $\bm{T}$ is independent of the perturbation $\bm{\Delta}$.
            Finally, we have to check condtion \eqref{eq:LemLLLOrderdCheckCond}.
            Because of \eqref{eq:LemmaLLLConstAsymR}, it holds that
        \begin{equation}
            \underset{\norm{\bm{\Delta}}_2 \rightarrow 0}{\lim} \left(\tilde{r}^2_{k-1,k-1}(\bm{\Delta}) -  \tilde{r}^2_{kk}(\bm{\Delta}) -  \tilde{r}^2_{k-1,k}(\bm{\Delta})\right) \le -  \tilde{r}^2_{k-1,k}(\bm{0}) 
        \end{equation}
        and we can always find a small enough $\norm{\bm{\Delta}}_2$ such that \eqref{eq:LemLLLOrderdCheckCond} is fulfilled and no column permutations occur.
        We have now completed a whole iteration of the LLL algorithm.
        Therefore, throughout the algorithm, the upper triangular structure of  $\bm{\tilde{R}}(\bm{\Delta})$ is preserved with $\tilde{r}_{ii} = r_{ii}(d_i + \Delta_i)$.
        Additionally, the property \eqref{eq:LemmaLLLConstAsymR} is fulfilled.
        Most importantly, the matrix $\bm{T}$ stays independent of the perturbation $\bm{\Delta}$ throughout the whole algorithm.
        This directly implies the initial claim.
        \end{proof}
   \end{lemma} 
}
\subsection{LLL-aided Nearest Plane}
We focus on the \ac{NP} algorithm where, intuitively, \ac{LR} should give an improvement to the conventional algorithm.
However, when using the \ac{NP} method with the optimized scaling \eqref{eq:NearestPlaneOptScaling}, we can show with Lemma \ref{lem:LLL_Ordered} that combining it with LLL reduction does not lead to a higher rate.
This is summarized in the following Theorem. 
\begin{theorem}\label{theo:LLLNPSameTHP}
    LLL-aided \ac{NP} leads to exactly the same solution as conventional \ac{NP} when considering the optimized scaling in \eqref{eq:NearestPlaneOptScaling}, i.e., 
    \begin{equation}
        \bm{a}^{\textnormal{LLL-NP}}=\bm{a}^{\textnormal{NP}}, \quad \text{and}  \quad R_{\textnormal{NP-LLL}}^{\textnormal{HSNR}} =    R_{\textnormal{NP}}^{\textnormal{HSNR}}.
    \end{equation}
  \end{theorem}
\begin{proof}
    We apply the LLL reduction with the unimodular integer matrix $\bm{T}$ (which in this case is upper-unitriangular, see Lemma \ref{lem:LLL_Ordered})
    \begin{equation}\label{eq:transmitwithpermutandLLL}
        \begin{aligned}
            \bm{x} &= \bm{H}^+\bm{D}^{{\star}}  \bm{T} (\bm{T}^{\inv}\bm{s} + \bm{T}^{\inv} \bm{a})
             =\bm{\tilde{Q}} \bm{\tilde{R}}(\bm{\tilde{s}} + \bm{\tilde{a}})
        \end{aligned}
    \end{equation}
    where $\bm{\tilde{s}} = \bm{T}^{\inv}\bm{s}$, $\bm{\tilde{a}} = \bm{T}^{\inv}\bm{a}$, and $\bm{\tilde{Q}} \bm{\tilde{R}}$ being the QR decomposition of $\bm{H}^+ \bm{D}^{{\star}}  \bm{T}$.
    According to \eqref{eq:SolutionNP} of the \ac{NP} algorithm, the integers $\tilde{a}_i$ are successively chosen as 
        $\tilde{a}_i = -\intround{\tilde{s}_i +\summe{j>i}{}{\tilde{r}_{ij}} (\tilde{s}_j + \tilde{a}_j)}$
    We assume now that all $\tilde{a}_j$ with $j>i$ have been obtained, and now $\tilde{a}_i$ is optimized.
    To this end, the matrices $\bm{T}$ and $\bm{R}$ are written as
    $\bm{T} = \left[\begin{smallmatrix}
        \bm{T}_{<i,\textnormal{a}} & \bm{T}_{<i,\textnormal{b}}\\
        \bm{0}&\bm{T}_i
    \end{smallmatrix}\right]$ and $\bm{R} = \left[\begin{smallmatrix}
        \bm{R}_{<i,\textnormal{a}} & \bm{R}_{<i,\textnormal{b}}\\
        \bm{0}&\bm{R}_i
    \end{smallmatrix}\right]$ where $\bm{T}_i$ and $\bm{R}_i$ are the upper-triangular blocks starting from the $i$-th row.
    Due to the triangular structure, we have $\bm{T}^{\inv} = \left[\begin{smallmatrix}
        \bm{T}_{<i,\textnormal{a}}^{\inv} & -\bm{T}_{<i,\textnormal{a}}^{\inv} \bm{T}_{<i,\textnormal{b}}\bm{T}_i^{\inv}\\
        \bm{0}&\bm{T}_i^{\inv}
    \end{smallmatrix}\right]$ and $\bm{R}\bm{T} = \left[\begin{smallmatrix}
        \bm{R}_{<i,\textnormal{a}}\bm{T}_{<i,\textnormal{a}} & \bm{R}_{<i,\textnormal{a}} \bm{T}_{<i,\textnormal{b}} +  \bm{R}_{<i,\textnormal{b}} \bm{T}_{i}\\
        \bm{0}&\bm{R}_i \bm{T}_i
    \end{smallmatrix}\right]$
    with the submatrices
    \begin{equation}\label{eq:TriangularTStructure}
        \bm{T}_i = \begin{bmatrix}
            1&\bm{t}_{i}^{\T}\\
            \bm{0}& \bm{T}_{i+1}
        \end{bmatrix},\quad
        \bm{T}_i^{\inv} = \begin{bmatrix}
            1&- \bm{t}_{i}^{\T}\bm{T}_{i+1}^{\inv}\\
            \bm{0}& \bm{T}^{\inv}_{i+1}
        \end{bmatrix},
    \end{equation}
    \begin{equation}\label{eq:TriangularRTStructure}
        \bm{R}_i = \begin{bmatrix}
            1&\bm{r}_{i}^{\T}\\
            \bm{0}& \bm{R}_{i+1}
        \end{bmatrix},\quad\text{and} \quad
       \bm{R}_i \bm{T}_i = \begin{bmatrix}
            1&\bm{t}_{i}^{\T}+\bm{r}_{i}^{\T} \bm{T}_{i+1}\\
            \bm{0}&\bm{R}_{i+1} \bm{T}_{i+1}
        \end{bmatrix}.
    \end{equation}
    As $\bm{T},\bm{T}^{\inv}$, and $\bm{R}$ are upper triangular matrices, only the blocks $\bm{T}_{i}$ and $\bm{R}_{i}$ have to be considered for optimzing $\tilde{a}_i$ and we arrive at
\begin{equation}\label{eq:NPScaledaitildereformulated}
    \abs{\bm{e}_i^{\T}\bm{RT}(\bm{\tilde{s}} + \bm{\tilde{a}}) }^2 =  \abs{\tilde{a}_i +\tilde{s}_i +(\bm{t}_{i}^{\T}+\bm{r}_{i}^{\T} \bm{T}_{i+1})(\bm{\tilde{s}}_{>i} + \bm{\tilde{a}}_{>i}) }^2.
\end{equation}
Minimizing this expression w.r.t. $\tilde{a}_i$ results in
\begin{equation}
    \tilde{a}_i = -\intround{ \tilde{s}_i +(\bm{t}_{i}^{\T}+\bm{r}_{i}^{\T} \bm{T}_{i+1})(\bm{\tilde{s}}_{>i} + \bm{\tilde{a}}_{>i})}.
\end{equation}
With $\bm{a} = \bm{T} \bm{\tilde{a}}$ and exploiting the matrix structures in \eqref{eq:TriangularTStructure} and \eqref{eq:TriangularRTStructure}, we have
\begin{equation}
    \begin{aligned}
    a_i &= \bm{t}_{i}^{\T}\bm{T}_{i+1}^{\inv}\bm{a}_{>i} -\intround{ \tilde{s}_i +(\bm{t}_{i}^{\T}+\bm{r}_{i}^{\T} \bm{T}_{i+1})(\bm{\tilde{s}}_{>i} + \bm{\tilde{a}}_{>i})}\\
    &= \bm{t}_{i}^{\T}\bm{T}_{i+1}^{\inv}\bm{a}_{>i} -\intround{ \tilde{s}_i +(\bm{t}_{i}^{\T}\bm{T}_{i+1}^{\inv}+\bm{r}_{i}^{\T} )(\bm{{s}}_{>i} + \bm{{a}}_{>i})}.
    \end{aligned}
\end{equation}
Because $\bm{t}_{i}^{\T}\bm{T}_{i+1}^{\inv}\bm{a}_{>i}$ is also an integer, we arrive at
\begin{equation}
\begin{aligned}   
    a_i &= -\intround{ \tilde{s}_i +\bm{t}_{i}^{\T}\bm{T}_{i+1}^{\inv}\bm{{s}}_{>i}+\bm{r}_{i}^{\T} (\bm{{s}}_{>i} + \bm{{a}}_{>i})}\\ 
    &= -\intround{ {s}_i +\bm{r}_{i}^{\T} (\bm{{s}}_{>i} + \bm{{a}}_{>i})}\\     
\end{aligned}
\end{equation}
and, hence, the same equation as in the normal \ac{NP} algorithm [cf. \eqref{eq:SolutionNP}] arises.
It follows that the LLL-aided version of the \ac{NP} algorithm results in exactly the same solution as the conventional \ac{NP} algorithm and, therefore, also results in the same performance. 
\end{proof}

\begin{corollary}\label{col:LLLTHPSameTHP}
    Theorem \ref{theo:LLLNPSameTHP} can also be derived for \ac{THP} by combining the \ac{THP} method according to \eqref{eq:THPSolution} with \ac{LR}.
    As the only difference between THP and the NP algorithm is a lower instead of an upper triangular matrix, the same concept of proof can be applied.
    Hence, also LLL-aided \ac{THP} with the scaling in \eqref{eq:THPOptScaling} cannot outperform \ac{THP} or NP.
\end{corollary}
\begin{remark}\label{rem:OnlyForScaled}
    Note that Theorem \ref{theo:LLLNPSameTHP} only holds for the optimized $\bm{D}^{{\star}}$ in \eqref{eq:NearestPlaneOptScaling}.
    For a different $\bm{D}$, LR can make a difference.
    In Section \ref{sec:Numerical}, we show that for $\bm{D} = \eye$, LR significantly improves upon the conventional \ac{NP} algorithm.
    However, this is no solution as the rate is clearly degraded in comparison to the optimized matrix in \eqref{eq:NearestPlaneOptScaling}.
\end{remark}

{
    This is further supported by the next proposition where we show that the optimized scaling in \eqref{eq:NearestPlaneOptScaling} is not only optimal for conventional \ac{NP} but also locally optimal for \ac{LLL}-aided \ac{NP}.
    \begin{proposition}\label{prop:LocalOptScalingLLLNP}
        The  optimized scaling in \eqref{eq:NearestPlaneOptScaling} is locally optimal for the \ac{LLL}-aided \ac{NP} algorithm.
    \end{proposition}
    \begin{proof}
    We are now proofing the optimality of the scaling \eqref{eq:NearestPlaneOptScaling} 
    \begin{equation}\label{eq:LemmaLocalOptNPScalingOptDefinition}
        \bm{d}^{\star} = [{1}/{r_{11}},{1}/{r_{22}},\dots,{1}/{r_{2K,2K}}]^{\T}.
    \end{equation}
    For the \ac{LLL}-aided \ac{NP} algorithm we have 
    \begin{equation}
        \begin{aligned}
                    \expct{\norm{\bm{\bar{x}}}^2} &= \summe{j=1}{2K} \tilde{r}^2_{jj} \mathbb{E}\bigg[\bigg|\textnormal{Mod}\bigg( \tilde{s}_j + \summe{m>j}{} \frac{\tilde{r}_{jm}}{\tilde{r}_{jj}}(\tilde{s}_m + \tilde{a}_m) \bigg)\bigg|^2\bigg].\\
        \end{aligned}
    \end{equation}
       Due to Lemma \ref{lem:ConstantLRMatrix} we know that $\bm{T}$ stays constant within a ball $\bm{d} \in B_{\delta}(\bm{d}^{\star})$ and, therefore, 
    \begin{equation}
        \begin{aligned}
                  & {\textnormal{Mod}\bigg( \tilde{s}_j + \summe{m>j}{} \frac{\tilde{r}_{jm}}{\tilde{r}_{jj}}(\tilde{s}_m + \tilde{a}_m) \bigg)} \\
                   &=  {\textnormal{Mod}\bigg( s_j +  \summe{k>j}{}t_{jk}^{\inv} + \summe{m>j}{} \frac{\tilde{r}_{jm}}{\tilde{r}_{jj}} \summe{n \ge m}{} (t_{mn}^{\inv}({s}_m + {a}_m) \bigg)}\\
                    &=  {\textnormal{Mod}( s_j +  w_j )}
        \end{aligned}
    \end{equation}
     where $t_{op}^{\inv}$ are the entries of the unitriangular $\bm{T}^{\inv}$ which are constant within $B_{\delta}(\bm{d}^{\star})$.
     As $w_j$ is independent of $s_j$, we have $\textnormal{Mod}( s_j +  w_j ) \sim \mathcal{U}(-0.5,0.5)$ due to Lemma \ref{lemma:ModuloStayUniform} and arrive at 
     \begin{equation}
            \expct{\norm{\bm{\bar{x}}}^2} = \frac{1}{12}\summe{j=1}{2K} \tilde{r}^2_{jj} = \frac{1}{12}\summe{j=1}{2K} {r}^2_{jj} d_j^2 \quad \forall \bm{d} \in B_{\delta}(\bm{d}^{\star})
     \end{equation}
     where we exploited that $\bm{\tilde{R}} = \bm{R}\bm{D}\bm{T}$ for an upper triangular $\bm{T}$. This is the exact same power as in case of the conventional \ac{NP} algorithm.
     Following Proposition \ref{prop:ScalingMatrixOptRate}, the optimal scaling within $B_{\delta}(\bm{d}^{\star})$ is given by \eqref{eq:LemmaLocalOptNPScalingOptDefinition} and, hence, the scaling \eqref{eq:NearestPlaneOptScaling} is locally optimal.
    \end{proof}
    \begin{remark}\label{rem:WhiteQuantNoise}
    It has to be noted that the optimal scaling \eqref{eq:NearestPlaneOptScaling} results in the noise $\expct{\bm{\tilde{R}}(\bm{\tilde{s}}+ \bm{\tilde{a}})(\bm{\tilde{s}}+ \bm{\tilde{a}})^{\T} \bm{\tilde{R}}^{\T}} = \frac{1}{12}\eye$  to be white for the (\ac{LLL})-aided \ac{NP} algorithm.
    In \cite{LatticeQuantizationNoiseWhite}, it has been shown that this is a necessary condition for the optimal quantizer.
\end{remark}
}

\subsection{LLL-aided Rounding Off}
Combining the \ac{RO} algorithm with \ac{LR} does not lead to the zero vector anymore.
While it improves on the conventional algorithm, also the LLL-aided \ac{RO} algorithm cannot exceed the rate of the conventional \ac{NP} (or THP) when considering the scaling in \eqref{eq:NearestPlaneOptScaling} [or \eqref{eq:THPOptScaling}] which is in accordance with the bounds in \cite{Babai}.
\begin{proposition}\label{prop:ROWorseNP}
    The LLL-aided \ac{RO} algorithm cannot exceed the performance of the conventional NP (or THP) when considering the optimized scaling in \eqref{eq:NearestPlaneOptScaling} [or \eqref{eq:THPOptScaling}].
\end{proposition}
\begin{proof}
    We start by using the transmit symbol as in \eqref{eq:transmitwithpermutandLLL} 
        \begin{equation}
        \begin{aligned}
            \bm{\bar{x}} &= \bm{H}^+ \bm{D}^{{\star}}   \bm{T}(\bm{\tilde{s}} + \bm{\tilde{a}}).
        \end{aligned}
    \end{equation}
    The RO procedure gives now the simple integer solution $\bm{\tilde{a}} = - \intround{\bm{\tilde{s}}}$,
    resulting in the symbol energy
    \begin{equation}
        \begin{aligned}
            \expct{\norm{\bm{\bar{x}}}^2} = \expct{\norm{\bm{H}^+\bm{D}^{{\star}}  \bm{T} \,\Mod{\tilde{\bm{s}}} }^2}. 
        \end{aligned}
    \end{equation}
    With $t_{ij}^{\inv}$ being the entries of $\bm{T}^{\inv}$, we can express $\Mod{\tilde{s}_i}$ as
    \begin{equation}
        \Mod{\tilde{s}_i } = \Mod{\bm{e}_i^{\T}\bm{T}^{\inv} \bm{s}} = \mathrm{Mod}\bigg(s_{i} + \summe{j> i}{} {t}^{\inv}_{ij} s_{j} \bigg)
    \end{equation}
    since $\bm{T}$ is unitriangular. Because $s_{i}$ is independent of $s_{j}$, we can conclude from Lemma \ref{lemma:ModuloStayUniform} that $\Mod{\tilde{s}_i}\sim\mathcal{U}(-0.5,0.5)$.
    Additionally, as $\tilde{s}_j$ only depends on $s_{m}, \,m \ge j$, and $s_{i}$ within $\tilde{s}_i$ is independent of ${s}_{j}$, $j>i$, it also follows that 
    \begin{equation}\label{eq:indepjgreati}
        \Mod{\tilde{s}_i }|\Mod{\tilde{s}_j } \sim \mathcal{U}(-0.5,0.5) \quad \text{for} \quad j >i.
    \end{equation}
    The reason is that $s_{i}$ within $\tilde{s}_i$ is independent of the condition, and with Lemma \ref{lemma:ModuloStayUniform}, we arrive at \eqref{eq:indepjgreati}.
    This means that $ \Mod{\tilde{s}_i }$ is independent of $ \Mod{\tilde{s}_j }$ for $ j> i$.
    As this holds $\forall j\neq i$,
    \begin{equation}
        \Mod{\tilde{s}_i } \, \text{is independent of}\,  \Mod{\tilde{s}_j } \quad \forall\, j\neq i.
    \end{equation}
    Therefore,
        $\expct{\bm{\tilde{s}}\bm{\tilde{s}}^{\T}} = \eye \frac{1}{12}$
    and the symbol energy is given by
    \begin{equation}
        \begin{aligned}
            \expct{\norm{\bm{H}^+\bm{D}^{{\star}}   \bm{T} \,\Mod{\tilde{\bm{s}}} }^2} =\tr\left(\bm{T}^{\T}\bm{D}^{{\star}}(\bm{H}\bm{H}^{\T})^{\inv} \bm{D}^{{\star}} \bm{T}\right)/12.
        \end{aligned}
    \end{equation}
    By using the arithmetic-geometric mean inequality, we obtain 
    \begin{equation}
        \begin{aligned}
            \expct{\norm{\bm{\bar{x}}}^2} &\ge 2K \sqrt[2K]{\det\left(\bm{T}^{\T} \bm{D}^{{\star}}(\bm{H}\bm{H}^{\T})^{\inv} \bm{D}^{{\star}}  \bm{T}\right)} \frac{1}{12}\\
            &= \frac{K}{6} \sqrt[2K]{\det\left((\bm{H}\bm{H}^{\T})^{\inv}   \right)}\sqrt[2K]{ \prod\nolimits_{i=1}^{2K} d_{ii}^2} .
        \end{aligned}
    \end{equation}
    Hence, we can upper bound the LLL-RO rate [see \eqref{eq:RateHSNR}] as
    \begin{equation}
        \begin{aligned}
            R_{\text{LLL-RO}}^{\text{HSNR}} &\le
             \log_2\det\left(\bar{p} \bH_{\mathbb{C}} \bH_{\mathbb{C}}^\Her \right) -   K\log_2\left(\frac{\pi e}{6}\right).\\
        \end{aligned}
    \end{equation}
    Note that the bound yields the same rate as for the conventional NP (or THP) algorithm.
    However, this upper bound is, in general, not achievable by the LLL-aided RO method.
\end{proof}{
\subsection{Dithering}
In \cite{DitheringRO}, a \ac{LR} method was analyzed where the transmit symbol is constructed as 
\begin{equation}
    \bm{\bar{x}} = \bm{H}^+\bm{D}\bm{T}\Mod{\bm{T}^{\inv}\bm{s} + \bm{v}}.
\end{equation}
In comparison to the \ac{LLL}-aided \ac{RO} algorithm, the dithering vector $\bm{v}$ is included which changes the system model as the receivers need to know $\bm{T}\bm{v}$.
The dithering vector is assumed to be independent and has entries distributed as $\mathcal{U}(-0.5,0.5)$
which results in the rows of $\Mod{\bm{T}^{\inv}\bm{s} + \bm{v}}$ being independent and distributed as $\mathcal{U}(-0.5,0.5)$ (see \cite{DitheringRO}).
Using the same $\bm{D}^{{\star}}$ and $\bm{T}$ as for Proposition \ref{prop:ROWorseNP}, this method has the exact same performance as the LLL-aided \ac{RO} algorithm.
In \cite{DitheringRO}, it was shown that this method is worse than \ac{THP}.
}

{
    This model drastically simplifies the analysis which allows us to show that there exists no integer matrix $\bm{T}$ with $\det(\bm{T})=1$ which improves on the simple \ac{NP} algorithm, rendering \ac{LLL} as unnecessary also for this model.
    We show this by using the transmit symbol
    \begin{equation}
        \bm{\bar{x}} = \bm{H}^+\bm{D}\bm{T}\left(\Mod{\bm{T}^{\inv}\bm{s} + \bm{v}} + \bm{\tilde{a}}\right)
    \end{equation}
    where $\bm{\tilde{a}}$ is obtained by the \ac{NP} algorithm. 
    Then,
    \begin{equation}
        \expct{\norm{\bm{\bar{x}}}^2} = \summe{i=1}{2K}\tilde{r}_{ii}^2 \frac{1}{12}
    \end{equation}
    due to the independence introduced by the dithering vector.
    Applying the arithmetic-geometric mean inequality, we have 
    \begin{equation}
        \begin{aligned}
            \expct{\norm{\bm{\bar{x}}}^2}%
            \ge\frac{K}{6} \sqrt[2K]{\det\left((\bm{H}\bm{H}^{\T})^{\inv}   \right)}\sqrt[2K]{ \prod\nolimits_{i=1}^{2K} d_{ii}^2} 
        \end{aligned}
    \end{equation}
    and arrive at the upper bound
    \begin{equation}
        \begin{aligned}
            R_{\text{Dither-LLL-NP}}^{\text{HSNR}} &\le
             \log_2\det\left(\bar{p} \bH_{\mathbb{C}} \bH_{\mathbb{C}}^\Her \right) -   K\log_2\left(\frac{\pi e}{6}\right)
        \end{aligned}
    \end{equation}
    which is the same rate as for the (\ac{LLL}-aided) \ac{NP} algorithm without dithering.\\
    While this approach simplifies the analysis considerably, it requires the assumption of the dithering model and does not allow a generalization.
    Hence, we won't further investigate it in this article.
}

{
\section{Special Lattice Structure}\label{sec:Generalization}
In the previous section, we have shown that the classical \ac{LR} techniques cannot outperform conventional \ac{THP} or the conventional \ac{NP} method.
This suggests that \ac{LR} can be omitted when considering these low-complexity algorithms.
In particular, the \ac{NP} algorithm leads to the exact same integer vector, regardless if LLL reduction is performed or not (Theorem~\ref{theo:LLLNPSameTHP}).
However, in the previous section, we used very specific algorithms like \ac{RO}, \ac{NP} for optimization and LLL for \ac{LR}.
In this section, we show that the results obtained in the last section can be considerably generalized.
This is possible due to a special lattice structure which will be analyzed in the following.
It is important to note that this structure is the only assumption that is made within this section.
This means that the following analysis can be applied to any lattice optimization where this structural property appears.
Specifically, we will show for this structure that \ac{LR} (beyond LLL) has no impact for a wide class of algorithms.
The special lattice structure, we are referring to, is defined as all lattices for which the basis matrix is in the set $\bm{B} \in \mathbb{U}$, i.e.,
\begin{equation}
    \begin{aligned}
        \mathcal{L}_{\mathbb{U}}&=\{\bm{B} \mathbb{Z}^{2K} | \bm{B} \in \mathbb{U}\},\\
        \unitriag &= \{\bm{B}=\bm{QR}: 
                         \abs{r_{ii}} \le \abs{r_{jj}} \,\text{for}\, i<j\}
    \end{aligned}
\end{equation}
where $\bm{B}=\bm{QR}$ is a QR-decomposition.
A unitriangular matrix $\bm{R}$ is the important special case which arises for the scaling in \eqref{eq:NearestPlaneOptScaling}.
\subsection{Generalization of Lattice Reduction}
At first, we focus on an observation within Theorem \ref{theo:LLLNPSameTHP}.
There, it has been observed that if the basis matrix of a lattice $\bm{B}$ lies in $\unitriag$, the integer matrix $\bm{T}$ of the reduced basis
\begin{equation}
    \bm{\tilde{B}} = \bm{BT} =\bm{\tilde{Q}}\bm{\tilde{R}} \in \unitriag
\end{equation}
is unitriangular (see Lemma \ref{lem:LLL_Ordered}) and that $\diag(\bm{\tilde{R}}) = \diag(\bm{R})$ holds.
That means that also $\bm{\tilde{B}}$ belongs to $\unitriag$.
This result has been shown in Lemma \ref{lem:LLL_Ordered} for the LLL reduction.
We will now generalize this result by showing that the LLL reduced matrix $\bm{\tilde{B}}$ is also \ac{HKZ} reduced.
\ac{HKZ} reduction belongs to the strongest lattice reductions and, in general, requires exponential runtime.
\begin{lemma}\label{lem:HKZ}
    If $\bm{B} \in \unitriag$, then the LLL-reduced lattice basis $\bm{\tilde{B}}~=~\bm{BT} \in \unitriag$ is also \ac{HKZ} reduced.
\end{lemma}
\begin{proof}
    A \ac{HKZ} basis is size reduced (see \cite[p.200]{HKZReduction}) which is automatically fulfilled for an LLL-reduced basis.
    In addition to that, for a \ac{HKZ} basis it has to hold that $\bm{\tilde{q}}_i \tilde{r}_{ii} = \bm{\tilde{Q}}\bm{e}_i\tilde{r}_{ii}$ is the shortest non-zero vector in $\pi_i(\mathcal{L}(\bm{\tilde{B}}))$ (see \cite[p.200]{HKZReduction}) where 
        $\pi_i(\bm{x}) = \summe{n=i}{2K} \bm{\tilde{q}}_n\bm{\tilde{q}}_n^{\T}\bm{x}$ 
    The lattice $\mathcal{L}(\bm{\tilde{B}})$ is represented by $\bm{\tilde{B}}\bm{a}$ where $\bm{a} \in \mathbb{Z}^{2K}$.
    Hence, we need to show that
    \begin{equation}
        \neg \exists\, \bm{a}\in \mathbb{Z}^{2K}\setminus \bm{0}:  \norm{\summe{n=i}{2K} \bm{\tilde{q}}_n\bm{\tilde{q}}_n^{\T}\bm{\tilde{B}}\bm{a}}^2 < \norm{\bm{\tilde{q}}_i \tilde{r}_{ii}}^2 =\abs{\tilde{r}_{ii}}^2,
    \end{equation}
    where the first term can be rewritten as 
    \begin{equation}
       \norm{\pi_i(\mathcal{L})}^2 =  \norm{\summe{n=i}{2K} \bm{\tilde{q}}_n\bm{\tilde{q}}_n^{\T}\bm{\tilde{B}}\bm{a}}^2 = \summe{n=i}{2K}\abs{ \bm{e}_n^{\T}\bm{\tilde{R}}\bm{a}}^2 = \summe{n=i}{2K}\abs{ \summe{l\ge n}{}\tilde{r}_{nl}a_l}^2
    \end{equation}
    where we, additionally, incorporated the triangular structure of $\bm{\tilde{R}}$ in the last term.
    The proof is performed by induction, where we assume that the integers $a_{m+1},a_{m+2},\dots,a_{2K}$ are equal to zero.
    With this assumption, only the rows of $i$ to $m$ are relevant due to the triangular structur of $\bm{\tilde{R}}$, i.e.,
    \begin{equation}
        \norm{\pi_i(\mathcal{L})}^2\Big|_{a_{l}=0,\, l>m} = \summe{n=i}{m-1}\abs{ \hspace*{-6pt}\summe{ n \le l\le m}{}\hspace*{-6pt}\tilde{r}_{nl}a_l}^2 + \abs{\tilde{r}_{mm}a_m}^2.
    \end{equation}
    Because $\bm{\tilde{B}} \in \unitriag$, we have 
    \begin{equation}
        \abs{\tilde{r}_{mm}a_m}^2 \ge  \abs{\tilde{r}_{ii}a_m}^2 \Rightarrow a_m = 0.
    \end{equation}
    Therefore, $a_m=0$ is necessary as $\norm{\pi_i(\mathcal{L})}^2|_{a_{l}=0,\, l>m} \ge \abs{\tilde{r}_{ii}}^2$ holds otherwise.
    To finish the proof by induction, we make no assumptions on $\bm{a}$.
Rewriting
    \begin{equation}
        \norm{\pi_i(\mathcal{L})}^2 = \summe{n=i}{2K-1}\abs{ \hspace*{-2pt}\summe{ l \ge n}{}\hspace*{-2pt}\tilde{r}_{nl}a_l}^2 + \abs{\tilde{r}_{2K,2K}a_{2K}}^2 
    \end{equation}
    it holds again that $ \abs{\tilde{r}_{2K,2K}a_{2K}}^2 \ge  \abs{\tilde{r}_{ii}a_{2K}}^2$ because $\bm{\tilde{B}} \in \unitriag$.
    With the same argumentation as above, $a_{2K} =0$ is necessary.
    Thus, it has been shown that only the zero vector leads to $\norm{\pi_i(\mathcal{L})}^2< \abs{\tilde{r}_{ii}}^2$ which is excluded by definition.
 \end{proof}
 Hence, also \ac{HKZ} reduction, which is a considerable more complex reduction, does not have any influence.
 \subsection{Generalization to{UTLR-invariant} Schemes}
Theorem \ref{theo:LLLNPSameTHP} is now generalized to a whole class of algorithms.
For this, we use the notation of Theorem  \ref{theo:LLLNPSameTHP}, with 
\begin{equation}
    \bm{T} = \begin{bmatrix}
        \bm{T}_{<i,\textnormal{a}} & \bm{T}_{<i,\textnormal{b}}\\
        \bm{0}&\bm{T}_i
    \end{bmatrix},\quad 
    \bm{R} = \begin{bmatrix}
        \bm{R}_{<i,\textnormal{a}} & \bm{R}_{<i,\textnormal{b}}\\
        \bm{0}&\bm{R}_i
    \end{bmatrix}
\end{equation}
 where $\bm{T}_i$ and $\bm{R}_i$ are the upper-triangular blocks starting from the $i$-th row.
 These are defined as
\begin{equation}
    \bm{T}_i = \begin{bmatrix}
        1&\bm{t}_{i}^{\T}\\
        \bm{0}& \bm{T}_{i+1}
    \end{bmatrix},\quad
    \bm{R}_i = \begin{bmatrix}
        r_{ii}&\bm{r}_{i}^{\T}\\
        \bm{0}& \bm{R}_{i+1}
    \end{bmatrix}.
\end{equation}
Please note that in comparison to  Theorem \ref{theo:LLLNPSameTHP}, we do not necessarily assume a unitriangular matrix $\bm{R}$ but $\bm{B} \in \unitriag$ and, therefore, $r_{ii}$ is in general not equal to one.
\begin{lemma}\label{lem:UTLRScheme}
    Let $\bm{B}$ be the basis matrix of a lattice $\mathcal{L}$ with QR-decomposition $\bm{B} = \bm{QR}$.
    Additionally, let $\bm{T}$ be a triangular, unimodular integer matrix (possibly from \ac{LR}).
    Furthermore, the shifted lattice vector is given by
    \begin{equation}\label{eq:SuccessiveSchemeAllDimensions}
        \begin{aligned}
           & \bm{Q}\bm{R}\bm{T}(\bm{\tilde{s}}+ \bm{\tilde{a}}) =  \bm{Q} \bm{R}(\bm{{s}}+ \bm{{a}}), \quad \text{with}\\
            &\bm{\tilde{s}}=\bm{T}^{\inv}\bm{s}\text{ and } \bm{\tilde{a}}=\bm{T}^{\inv}\bm{a}
        \end{aligned}
    \end{equation}
    where  $\bm{Q} \bm{R} \bm{{a}} \in \mathcal{L}$ is the variable of some lattice problem.
    Moreover, an optimization scheme is assumed which, in the $i$-th step, only considers the dimensions $\bm{Q}_i = [\bm{q}_i,\dots,\bm{q}_{2K}]$.\\
    Then, due to the triangular structure of $\bm{T}$, it is possible to form an equivalent formulation to \eqref{eq:SuccessiveSchemeAllDimensions}  within the sublattice $\mathcal{L}_i$ of $\mathcal{L}$ as
    \begin{equation}\label{eq:SuccessiveSchemeSublattice}
        \begin{aligned}
           & \bm{Q}_i \bm{R}_i\bm{T}_i(\bm{\tilde{s}}_{\ge i}+ \bm{\tilde{a}}_{\ge i}) =  \bm{Q}_i \bm{R}_i(\bm{{s}}_{\ge i}+ \bm{{a}}_{\ge i}) \quad \text{with} \\
            &\bm{\tilde{s}}_{\ge i}=\bm{T}_i^{\inv}\bm{s}_{\ge i}\text{ and } \bm{\tilde{a}}_{\ge i}=\bm{T}_i^{\inv}\bm{a}_{\ge i}.
        \end{aligned}
    \end{equation}
    where $\bm{Q}_i \bm{R}_i\bm{{a}}_{\ge i} \in \mathcal{L}_i$. If the optimization scheme now gives the same set of potential solution vectors
    \begin{equation}\label{eq:SublatticeSameIntegerSolution}
       \left\{\bm{a}_{\ge i}^{\star,(m)}\right\}_{m=1}^{M_i} =  \left\{\bm{T}_i \bm{\tilde{a}}_{\ge i}^{\star,(m)}\right\}_{m=1}^{M_i}
    \end{equation}
    in the $i$-th step after optimizing w.r.t. $\bm{\tilde{a}}_{\ge i}$ as when optimizing w.r.t.  $\bm{a}_{\ge i}$ directly, the matrix $\bm{T}$ has no impact on the problem.
\end{lemma}
\begin{proof}
    This follows directly from the triangular structure of the matrix $\bm{T}$.
\end{proof}
\begin{remark}
    In general, \ac{LR} is only unnecessary if the optimal solution of the lattice problem is found.
    However, if $\bm{B} \in \mathbb{U}$, it is enough if the optimal solution is found in each subproblem to render \ac{LR} as unnecessary according to Lemma \ref{lem:UTLRScheme}.
\end{remark}
{
\begin{figure}[t!]
    \hspace*{0.2cm}
    \subfigure[$\bm{B} \notin \mathbb{U}$]{
                \hspace{0.5cm}
    \includegraphics*{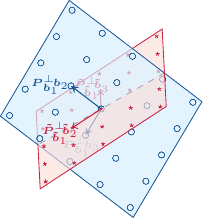}\\
    \label{fig:LatticeIllustrationNotU}
    }\hspace*{0.7cm}\subfigure[$\bm{B} \in \mathbb{U}$]{
        \hspace{-1.8cm}
        \includegraphics*{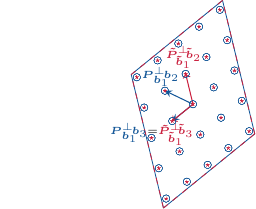}
            \label{fig:LatticeIllustrationU}
    }
    \caption{Illustration of the 2D sublattices with and without \ac{LR} depending if the lattice basis matrix of the 3D lattice is within the set $\mathbb{U}$ or not.}
    \label{fig:LatticeIllustration}
\end{figure}
    The invariance w.r.t. \ac{LR} is illustrated in Fig. \ref{fig:LatticeIllustration}. 
    In Fig.~\ref{fig:LatticeIllustrationNotU} the lattice basis matrix lies not in the set $\mathbb{U}$. 
    Then, the sublattices with \ac{LR} and without \ac{LR} within a subproblem are different and even if the optimal solution is found in each of them, these solutions are different.
    On the other hand, if the lattice basis matrix lies within the set $\mathbb{U}$ we arrive at Fig. \ref{fig:LatticeIllustrationU}.
    While the basis is different, the sublattice is the same for $\bm{B} \in \mathbb{U}$ in each subproblem regardless if \ac{LR} is applied or not.
    Hence, if the optimal solution is found within the subproblem, \ac{LR} makes no difference.
}
\begin{definition}
    We call an optimization scheme which complies with the assumptions of Lemma \ref{lem:UTLRScheme} and which fulfills \eqref{eq:SublatticeSameIntegerSolution} "\ac{UTLR}-invariant".
\end{definition}
\begin{corollary}\label{col:ScalarUTLRScheme}
    In case of a successive scalar optimization, the variables $\bm{\tilde{a}}_{> i}= \bm{T}_{i+1}^{\inv} \bm{{a}}_{> i}$ are fixed in the $i$-th step and only the variable ${\tilde{a}}_{i}$ is optimized.
    Hence, \eqref{eq:SuccessiveSchemeSublattice} reduces to 
    \begin{equation}
    \begin{aligned}
        & \bm{Q}_i \bm{R}_i\begin{bmatrix}
            s_i + \tilde{a}_i +  \bm{t}_{i}^{\T}\bm{T}_{i+1}^{\inv} \bm{a}_{>i} \\
            \bm{s}_{>i} + \bm{a}_{>i}
        \end{bmatrix} = \bm{Q}_i \bm{R}_i\begin{bmatrix}
            s_i + a_i \\
            \bm{s}_{>i} + \bm{a}_{>i}
        \end{bmatrix},\\
         &  \tilde{a}_i = {a}_i  - \bm{t}_{i}^{\T}\bm{T}_{i+1}^{\inv} \bm{a}_{>i} .
     \end{aligned}
 \end{equation}
    If the successive scalar optimization scheme gives the same set of potential solutions
    \begin{equation}
        \begin{aligned}
            \left\{{a}_i^{\star,(l)}\right\}_{l=1}^{L_i}= \left\{\tilde{a}_i^{\star,(l)}  + \bm{t}_{i}^{\T}\bm{T}_{i+1}^{\inv} \bm{a}_{>i} \right\}_{l=1}^{L_i}
        \end{aligned}
    \end{equation}
    in the $i$-th step for any fixed $\bm{a}_{>i} $ after optimizing w.r.t. ${\tilde{a}}_{ i}$ as when optimizing w.r.t. ${a}_{ i}$ directly, the scheme is{UTLR-invariant}.
    Please note that $\bm{a}_{>i}$ and, therefore, also  $\bm{\tilde{a}}_{> i}= \bm{T}_{i+1}^{\inv} \bm{{a}}_{> i}$ are fixed (to a candidate vector of a previous step). 
\end{corollary}
In this scalar case, the requirement is only w.r.t. a simple integer offset and almost all successive scalar optimization schemes are{UTLR-invariant}.
With the definition of the successive schemes, we can provide the following Theorem.
\begin{theorem}\label{theo:GeneralGreedy}
    Let $\bm{B}  \in \unitriag$ be the basis matrix of a lattice $\mathcal{L}$ with QR-decomposition $\bm{B} = \bm{QR}$.
    Additionally, let $\smash{\scalebox{0.8}{$\underset{\bm{\tilde{a}}\in \mathbb{Z}^{2K}}{\min}$}}\,\norm{\bm{QR}\bm{T}(\bm{\tilde{s}} + \bm{\tilde{a}})}_2$ with $\bm{\tilde{s}}=\bm{T}^{\inv}\bm{s}$ and $\bm{\tilde{a}}=\bm{T}^{\inv}\bm{a}$ be the transformed lattice problem after applying \ac{LR} with the unimodular integer matrix $\bm{T}$.
     Then, \ac{LR} has no impact if an{UTLR-invariant} scheme according to Lemma \ref{lem:UTLRScheme} and Corollary~\ref{col:ScalarUTLRScheme} is used.
\end{theorem}
\begin{remark}
    Please note that Theorem \ref{theo:GeneralGreedy} also holds for different optimization problems.
\end{remark}
An important special case is the scalar successive scheme.
Minimizing the two-norm in Theorem  \ref{theo:GeneralGreedy} in the $i$-th step of a successive scalar scheme w.r.t. $\tilde{a}_i \in \mathbb{\tilde{A}}_i= \mathbb{A}_i -\bm{t}_{i}^{\T}\bm{T}_{i+1}^{\inv} \bm{a}_{>i} $  
\begin{equation}\label{eq:Theorem2ScalarOpt}
    \begin{aligned}
        &\tilde{a}_i^{\star} = \underset{\tilde{a}_i \in \mathbb{\tilde{A}}_i}{\argmin} \norm{\bm{Q}_i\bm{R}_i\bm{T}_i(\bm{\tilde{s}} + \bm{\tilde{a}})}_2\\
        &= \underset{\tilde{a}_i \in \mathbb{\tilde{A}}_i}{\argmin} \abs{r_{ii}(s_i+\tilde{a}_i + \bm{t}_{i}^{\T}\bm{T}_{i+1}^{\inv}\bm{a}_{>i}) +\bm{r}_{i}^{\T} (\bm{{s}}_{>i} + \bm{{a}}_{>i})}\\
        &= -\bm{t}_{i}^{\T}\bm{T}_{i+1}^{\inv} \bm{a}_{>i} + \underset{\Delta\tilde{a}_i \in \mathbb{{A}}_i}{\argmin} \abs{r_{ii}(s_i +\Delta\tilde{a}_i)+\bm{r}_{i}^{\T} (\bm{{s}}_{>i} + \bm{{a}}_{>i})}\\
        &= -\bm{t}_{i}^{\T}\bm{T}_{i+1}^{\inv} \bm{a}_{>i} + a_i^{\star}
    \end{aligned}
\end{equation}
 gives the same solution as when minimizing via the original variable $a_i \in \mathbb{A}_i$.
Additionally, if more than one candidate is considered in each step of the scheme then also any new integer $\tilde{a}_i^{\prime} = \tilde{a}_i + \Delta$ due to an offset $\Delta \in \mathbb{Z}$ directly corresponds to the neighbor ${a}_i^{\prime} = {a}_i + \Delta $ with  $\bm{a}_{>i}$ unchanged.
}

{
Hence, most methods which work in a successive manner are{UTLR-invariant} and are not affected by \ac{LR} according to Theorem \ref{theo:GeneralGreedy}.
Popular examples are now categorized in the following.}
\subsubsection{Nearest Plane/THP ({UTLR-invariant})}
The \ac{NP} algorithm, or the \ac{THP} algorithm, are probably the most popular{UTLR-invariant} schemes.
The \ac{NP} algorithm starts with the integer $\tilde{a}_{2K}$ and then successively optimzies each scalar integer $\tilde{a}_i$ until $\tilde{a}_{1}$ is reached {while fixing the previously obtained integers (see Theorem \ref{theo:LLLNPSameTHP} for more details)}.
{This complies with an UTLR-invariant scheme and, according to Theorem \ref{theo:GeneralGreedy}, if $\bm{B} \in \mathbb{U}$, this approach results in the same solution regardless of whether \ac{LR} is applied or not.}
\subsubsection{K-Best Algorithm ({UTLR-invariant})}
The K-Best algorithm (see, e.g., \cite{KBestVLSI,KBestImpl,KBestImp,FPGAQRDMFSD,KBestLLLOne,KBestLLLTwo,KBestLLLThree}) is also an{UTLR-invariant} scheme and generalizes the concept of the simple \ac{NP} algorithm.
Instead of just successively taking the best integer in each step (as the \ac{NP} or \ac{THP} algorithm does), the $B$ best points are considered.
These K-Best type of algorithms have already been shown to achieve very good performance close to optimality.
However, for optimized performance, it is important that \ac{LR} is performed (see \cite{KBestLLLOne,KBestLLLTwo,KBestLLLThree}).
The \ac{LR}-aided K-Best algorithm starts with the integer $\tilde{a}_{2K}$ and determines the $B$ best candidates.
Afterward, for each candidate ${\tilde{a}}_{2K}^{(b)}$, $b=1,\dots,B$, the $C$ best children $\tilde{a}_{2K}^{(b),1}, \dots, \tilde{a}_{2K}^{(b),C}$ are derived in a greedy manner.
Then again, the best $B$ candidates are taken of all $BC$ children and the next layer ($2K-1$) is considered with the candidates $\bm{\tilde{a}}_{2K-1}^{(b)}$, $b=1,\dots,B$.
Similarly, the $C$ best children of the $B$ candidates are calculated, and the procedure is continued until the final integer $\tilde{a}_{1}$ is considered.
Lastly, the best candidate is taken and the solution vector  $\bm{\tilde{a}}^{\star}$ is obtained.
This K-Best procedure is an{UTLR-invariant} scheme, and we can provide the following corollary.
\begin{corollary}
      {If $\bm{B} \in \unitriag$, the \ac{LR}-aided K-Best algorithm leads to exactly the same solution as the conventional K-Best algorithm.
    $\bm{B} \in \unitriag$ appears for the optimized scaling in \eqref{eq:NearestPlaneOptScaling}.}
\end{corollary}
\begin{proof}
    {We assume that we have given the candidates $\bm{\tilde{a}}_{i+1}^{(b)}, b=1,\dots,B$ in the $i+1$-th layer and we are searching for the $B$ best candidates for in the $i$-th layer of a \ac{LR}-aided K-Best algorithm.}
    The candidates of the previous layer $\bm{\tilde{a}}_{i+1}^{(b)}$, {$b=1,\dots,B$} directly correspond to $\bm{{a}}_{i+1}^{(b)}, b=1,\dots,B$, i.e., the $B$ candidates of a conventional K-Best algorithm without \ac{LR}.
    Now, the best $C$ children of each candidate $\bm{\tilde{a}}_{i+1}^{(b)}$ have to be derived. 
    According to \eqref{eq:Theorem2ScalarOpt}, optimizing $\norm{\bm{e}_i^{\T}\bm{R}\bm{T}(\bm{\tilde{s}} + \bm{\tilde{a}})}$ for a given $\bm{\tilde{a}}_{i+1}^{(b)}$ w.r.t. $\tilde{a}_i$ results in the same integer ${a}_i$ when maximizing $\norm{\bm{e}_i^{\T}\bm{R}(\bm{{s}} + \bm{{a}})}$ for the given $\bm{{a}}_{i+1}^{(b)}$ without \ac{LR}.
    Hence, the optimized integer value $\tilde{a}_i^{(b),1}$, directly corresponds to ${a}_i^{(b),1}$ which is the optimized integer value without \ac{LR}.
    {Similarly, the $c$-th best child is determined by}
    \begin{equation}
        \tilde{a}_i^{(b),c} = \underset{\tilde{a}_i \in \mathbb{Z} \setminus \left\{\tilde{a}_i^{(n),1},\dots,\tilde{a}_i^{(b),c-1}\right\} }{{\argmin}} \abs{ \bm{e}_i^{\T}\bm{R}\bm{T}(\bm{\tilde{s}} + \bm{\tilde{a}})}^2.
    \end{equation}
    {and} the set of the children $\{\tilde{a}_i^{(b),1}, \dots, \tilde{a}_i^{(b),C}\}$ directly corresponds to the set $\{{a}_i^{(b),1}, \dots, {a}_i^{(b),C}\}$ which is obtained by the same procedure without \ac{LR}.
    Accordingly, also determining the $B$ best of all $BC$ children results in the same {set of candidates $\bm{\tilde{a}}_{i}^{(b)}, b=1,\dots,N$ regardless if \ac{LR} is applied or not.}
    {After} successively proceeding with the algorithm, the same solution vector is obtained for the \ac{LR}-aided K-Best and conventional K-Best algorithm.
\end{proof}

\subsubsection{Sphere Decoder (UTLR-invariant)}
The \ac{SD} is a very popular algorithm, see e.g. \cite{SphereDecoder,SphereDecoderFinckePohst,SphereDecoderDetection,SphereDecoderVLSI,SphereDecoderExpCompl}.
As we are dealing with precoding in the \ac{DL}, it is also called the sphere encoder (see \cite{VP}). 
While the \ac{SD} gives the optimal solution of the lattice problem, its exponential complexity renders it impractical.

{ The \ac{SD} algorithm according to \cite{SphereDecoder,SphereDecoderFinckePohst} is considered with the improved candidate selection according to \cite{SchnorrEuchner,ClosestLatticeSearch}.}
While not obvious at first, the \ac{SD} also belongs to the{UTLR-invariant} schemes and we will analyze this in the following.
As the \ac{SD} is optimal, it is apparent that with or without \ac{LR}, the same optimal point is reached.
The complexity, on the other hand, is in general reduced when \ac{LR} is applied (see \cite{LLLAnalysis}).
However, according to \cite{LLLAnalysis}, when no column permutations occur during \ac{LR}, the complexity is, in general, not reduced.
The complexity analysis in \cite{LLLAnalysis} is based on an approximation.
However, for our special case {of $\bm{B} \in \mathbb{U}$}, we can extend the analysis in \cite{LLLAnalysis} and show that the same sequence of lattice points is visited regardless of whether \ac{LR} is applied or not.
Hence, for the problem we consider, the complexity is actually exactly the same.
\begin{algorithm}
    \caption{Sphere Decoder \cite{SphereDecoder,SphereDecoderFinckePohst}}\label{alg:Spheredecoder}
    \begin{algorithmic}
    \Require $r > 0$
    \State Initialize: $i=2K$
    \begin{enumerate}
        \item For fixed $\bm{\tilde{a}}_{>i}=[\tilde{a}_{i+1},\dots,\tilde{a}_{2K}]^{\T}$, find all $\tilde{a}_i \in \mathbb{A}_i $ for which $\norm{\bm{R}_i\bm{T}_i(\bm{\tilde{s}}_i + \bm{\tilde{a}}_i)}^2 \le r^2$
        \If{$\mathbb{\tilde{A}}_i = \emptyset$}
        \State $i \gets i+1$
        \If{$i=2K+1$}
        \State Return($\bm{\tilde{a}}^{\star}$)
        \Else 
        \State GoTo 2)
        \EndIf
        \EndIf
        \State Pick $\tilde{a}_i \in \mathbb{\tilde{A}}_i $ with smallest $\norm{\bm{R}_i\bm{T}_i(\bm{\tilde{s}}_i + \bm{\tilde{a}}_i)}^2$
        \State Remove $\tilde{a}_i$ from $\mathbb{\tilde{A}}_i $
        \State $i \gets i-1$
        \If{$i\ge 1$}
            \State GoTo 1)
        \Else 
            \State Update $\bm{\tilde{a}}^{\star}$
            \State Update $r$
            \State $i=2$
            \State GoTo 2)
        \EndIf

    \end{enumerate}
    \end{algorithmic}
    \end{algorithm}

\begin{corollary}
    {If $\bm{B} \in \unitriag$, the \ac{LR}-aided \ac{SD} algorithm visits exactly the same lattice points in the same order as the conventional \ac{SD} algorithm.
    $\bm{B} \in \unitriag$ appears for the optimized scaling in \eqref{eq:NearestPlaneOptScaling}.}
\end{corollary}
\begin{proof}
    While the \ac{SD} achieves the optimum performance, it is still an{UTLR-invariant} procedure in the sense that the search tree is always successively constructed.
    {This can be shown similarly to the algorithms above and we give a short proof sketch in the following. 
    To see this, we show that the same candidates are generated in line~1) of Algorithm \ref{alg:Spheredecoder} and also selected in line~10).
    In line 1), the set $\mathbb{\tilde{A}}_i$ is constructed by determining all $\tilde{a}_i$ for which $\abs{\bm{e}_i^{\T}\bm{R}_i\bm{T}_i(\bm{\tilde{s}}_i + \bm{\tilde{a}}_i)}^2 \le r^2$ holds (for fixed $\bm{\tilde{a}}_{>i}=[\tilde{a}_{i+1},\dots,\tilde{a}_{2K}]^{\T}$).
    As $\tilde{a}_i$ corresponds to $ a_i= \tilde{a}_i + \bm{t}_{i}^{\T}\bm{T}_{i+1}^{\inv}\bm{a}_{>i} $ it directly follows that
    { \begin{equation}
        \begin{aligned}
            &\abs{\bm{e}_i^{\T}\bm{R}\bm{T}(\bm{\tilde{s}} + \bm{\tilde{a}})}^2\hspace*{-2pt}\\
            &  = \hspace*{-3pt} \abs{r_{ii}(\tilde{a}_i + \bm{t}_{i}^{\T}\bm{T}_{i+1}^{\inv}\bm{a}_{>i} + {s}_i) +\bm{r}_{i}^{\T} (\bm{{s}}_{>i} + \bm{{a}}_{>i})}^2 \le r^2 \\
            &\iff \abs{ r_{ii}({a}_i +  {s}_i) +\bm{r}_{i}^{\T} (\bm{{s}}_{>i} + \bm{{a}}_{>i})}^2 \le r^2
        \end{aligned}
    \end{equation}}
    and, hence, the set $\mathbb{\tilde{A}}_i=\{{a}_i^{(1)}+\bm{t}_{i}^{\T}\bm{T}_{i+1}^{\inv} \bm{a}_{>i},\dots,{a}_i^{(N)}+\bm{t}_{i}^{\T}\bm{T}_{i+1}^{\inv} \bm{a}_{>i} \}$ for which $\abs{\bm{e}_i^{\T}\bm{R}\bm{T}(\bm{\tilde{s}} + \bm{\tilde{a}})}^2\le r^2$ directly corresponds to the set $\mathbb{A}_i=\{{a}_i^{(1)},\dots,{a}_i^{(N)}\}$ of the \ac{SD} without \ac{LR} for which $\abs{\bm{e}_i^{\T}\bm{R}(\bm{{s}} + \bm{{a}})}^2 \le r^2$.
    Besides the generation also the same candidates are selected.
    In line 10), $\tilde{a}_i \in \mathbb{\tilde{A}}_i$ is chosen which minimizes the objective function based on a greedy step.
    This corresponds to $\tilde{a}_i^{\star} =   \argmin{\tilde{a}_i \in \mathbb{\tilde{A}}_i} \abs{ \bm{e}_i^{\T}\bm{R}\bm{T}(\bm{\tilde{s}} + \bm{\tilde{a}})}^2$.
    From \eqref{eq:Theorem2ScalarOpt}, we know that this is equivalent to ${a}_i^{\star} = {\argmin}_{{a}_i \in \mathbb{A}_i} \abs{ \bm{e}_i^{\T}\bm{R}(\bm{{s}} + \bm{{a}})}^2$.
    In summary, the same candidates are determined, and the same lattice points are selected.
    Hence, the \ac{LR}-aided \ac{SD} and the conventional \ac{SD} without \ac{LR} visit the same lattice points.}
\end{proof}

{
Many popular algorithms are{UTLR-invariant} and comply with Theorem \ref{theo:GeneralGreedy}.
However, there exist algorithms for which \ac{LR} does make a difference even if $\bm{B} \in \unitriag$. 
These will be examined in the following.}
\subsubsection{Rounding Off (Not{UTLR-invariant})}
We have already analyzed the Rounding Off procedure.
The \ac{RO} algorithm does not belong to the{UTLR-invariant} methods. 
Without \ac{LR}, the solution is the zero-vector according to \eqref{eq:ROZeros}, whereas when considering \ac{LR}, the optimal solution is, in general, an integer vector different from zero.
This also holds for {$\bm{B} \in \mathbb{U}$}, which arises for lattice precoding when considering the optimized scaling in \eqref{eq:NearestPlaneOptScaling}.
However, we also know from Proposition \ref{prop:ROWorseNP} that the performance cannot exceed the \ac{NP} or \ac{THP} algorithm even when \ac{LR} is applied.
\subsubsection{ {Constrained} Exhaustive Search (Not{UTLR-invariant})}
{Constrained} exhaustive search is another example of an algorithm not following a{UTLR-invariant} procedure.
For {constrained} exhaustive search, we restrict each integer value $\tilde{a}_i$ to lie within an interval around some intial point $\tilde{a}_i^{(\textnormal{init})}$ which results in the constraint set $-l_i \le \tilde{a}_i -\tilde{a}_i^{(\textnormal{init})}\le u_i$.
Afterward, all possible combinations are considered, and the best one is taken.
For this method, \ac{LR} does make a difference in general.
\begin{corollary}
    {A \ac{LR}-aided constrained exhaustive search algorithm generally leads to a different solution than a conventional constrained exhaustive search, even if $\bm{B} \in \unitriag$. (Excluding the case where the search space is unconstrained and the optimal solution is obtained.)}
\end{corollary}
\begin{proof}
    Firstly, we consider {another lattice} point which lies within the constraint set by changing the $i$-th integer as
    \begin{equation}
       \tilde{a}_i = \tilde{a}_i^{(\textnormal{init})} + \Delta
    \end{equation}
    where $\Delta \in \{\pm1, \pm2,\dots\}$ such that $-l_i \le \Delta \le u_i$.
    {This corresponds directly to the integer}
    \begin{equation}
        \begin{aligned}
                {a}_i &= {a}_i^{(\textnormal{init})} + \Delta\\
        \end{aligned}
    \end{equation}
    of ${a}_i^{(\textnormal{init})}$ without \ac{LR} with $\bm{a}_{>i}$ unchanged.
    However, the integer values $\bm{a}_{<i}$ are implicitly also changed by adding $\Delta$ to ${a}_i^{(\textnormal{init})}$, which results from the exhaustive search not belonging to the{UTLR-invariant} procedures.
    Actually, the integers $a_j$ with $j<i$ are given as
    \begin{equation}
        \begin{aligned}
            a_j &= \tilde{a}_j^{\text{(init)}} + \bm{t}_{j}^{\T} \bm{\tilde{a}}_{>j}\\
            &=\tilde{a}_j^{\text{(init)}} + \bm{t}_{j}^{\T}( \bm{\tilde{a}}_{>j}^{(\textnormal{init})} + \bm{e}_{i-j} \Delta)\\
            &={a}_j^{\text{(init)}} + t_{ji} \Delta.
        \end{aligned}
    \end{equation}
    As $t_{ji}$ is an arbitrary integer and can have any value, {${a}_j^{\text{(init)}} + t_{ji} \Delta$} might lie outside the constraint set $-l_i \le {a}_i -{a}_i^{(\textnormal{init})}\le u_i$ of the conventional exhaustive search without \ac{LR}.
   \end{proof}

\subsubsection{Fixed-Complexity Sphere Decoder (Not{UTLR-invariant})}
Due to the exponential complexity, exhaustive search is typically not employed for lattice problems.
However, there are also more practical methods where \ac{LR} still has an effect.
One of them is the \ac{FSD} \cite{FSD,FSDTwo,FSDThree,FSDFour}, which showed very good performance for the detection problem.
This approach has also been applied for precoding (see \cite{FSE,FSETwo,FSEThree}), where it is also called the Fixed-Complexity Sphere Encoder.
For the \ac{FSD}, the number of nodes in each layer is fixed, where the specific number is a design choice.
Typically, the integers $N+1,\dots,2K$ are determined by exhaustive search and then, for all these candidates, the \ac{NP} algorithm is applied for the remaining layers (see \cite{FSD,FSE}).
We observed that the optimal lattice point is typically very close to the \ac{LLL}-\ac{RO} or the (\ac{LLL})-\ac{NP} solution vector.
Hence, we restrict the search space of each element close to an initial solution (like the \ac{NP} solution), i.e., $ -\Delta \le a_i-a_i^{(\text{init})} \le \Delta$.
By choosing now a smaller number for $\Delta$, the number of integers for exhaustive search that can be computationally afforded increases, and we have observed a clearly improved performance in our scenarios.
As the layers $N+1,\dots,2K$ are determined by exhaustive search, the \ac{FSD} does not belong to the{UTLR-invariant} methods, which we summarize in the following corollary.
\begin{corollary}
    {A \ac{LR}-aided \ac{FSD} generally leads to a different solution than a conventional \ac{FSD} search, even if $\bm{B} \in \unitriag$.
    (Excluding the case where the search space of the exhaustive search is unconstrained and the optimal solution is obtained within these layers.)}
\end{corollary}
\begin{proof}
As the first layers $N+1,\dots,2K$ are determined by exhaustive search, the \ac{FSD} generally leads to a different solution when \ac{LR} is applied.
\end{proof}

{
\subsection{Generalization to Rate Region}
Until now, we have considered the sum rate where all user rates are incorporated equally.
However, changing to a rate region perspective, the weighted sum rate
\begin{equation}
    R^{\text{HSNR}}_{\text{WSR}} 
    = \summe{k=1}{K}\mu_k^{\prime}  R^{\text{HSNR}}_{k} =  \frac{1}{2}\summe{k=1}{2K}\mu_k\log_2 \left( \frac{P_{\text{Tx}}}{\pi e \expct{\norm{\bm{\bar{x}}}^2}} d_k^2\right)
\end{equation}
with $\summe{j=1}{K}\mu_j^{\prime}=K$ has to be evaluated.
We will show now that the results of this article can be generalized to the complete rate region by considering the weighted sum rate.
\begin{lemma}\label{lem:WSROptimalScaling}
    Given a fixed encoding order, the optimal diagonal scaling matrix $\bm{D}^{{\star}}$ for the \ac{NP} algorithm, considering the weighted sum rate, is 
    \begin{equation}
        \bm{D}^{{\star}} = \diag^{\inv}(\bm{r})\diag(\bm{\sqrt{\mu}})
    \end{equation}
    where $\bm{r}$ is the diagonal of $\bm{R}$ from the QR decomposition $\bm{H}^+ = \bm{QR}$ and $\bm{\sqrt{\mu}} = [\sqrt{\mu_1},\dots,\sqrt{\mu_{2K}}]^{\T}$.
\end{lemma}
\begin{proof}
    Similar to Proposition \ref{prop:ScalingMatrixOptRate}, we  have $\expct{\norm{\bm{\bar{x}}}^2}=\frac{1}{12}\summe{m=1}{2K} {r}_{mm}^2 d_m^2$ resulting in the weighted sum rate
    \begin{equation}
        R^{\text{HSNR}}_{\text{WSR}} =  \frac{1}{2}\summe{k=1}{2K}\mu_k\log_2 \left( \frac{12 P_{\text{Tx}}}{\pi e \summe{m=1}{2K} {r}_{mm}^2 d_m^2} d_k^2\right).
    \end{equation}
    To obtain the optimal scaling, we take the first-order derivative
    \begin{equation}
       \frac{\partial}{\partial d_j^2} R^{\text{HSNR}}_{\text{WSR}} = \frac{1}{2 \ln(2)}\mu_j \frac{1}{d_j^2} - \frac{2K}{2\ln(2)} \frac{r_{jj}^2}{\summe{m=1}{2K} {r}_{mm}^2 d_m^2}. 
    \end{equation}
    Setting the derivative to zero results in the optimal scaling
    \begin{equation}\label{eq:OptimalScalingWSR}
        d_j^{2,\star} = \frac{\mu_j}{r_{jj}^2} \left(\frac{1}{2K}\summe{j=1}{2K} {r}_{mm}^2 d_m^2\right) \propto \frac{\mu_j}{r_{jj}^2}
    \end{equation}
    where we can choose any proportionally factor as the objective is scaling invariant. 
\end{proof}
Using the optimal scaling of \eqref{eq:OptimalScalingWSR}, the objective can now be written as 
\begin{equation}\label{eq:WSRRateForOptimalScaling}
    R^{\text{HSNR}}_{\text{WSR}} =  \frac{1}{2}\summe{k=1}{2K}\mu_k\log_2 \left( \frac{6}{\pi e}\frac{ \mu_k}{r_{kk}^2}\bar{p}\right).
\end{equation}   
\begin{lemma}\label{lem:OptimizedEncodingOrder}
    The encoding order is optimal if the streams are encoded in ascending order w.r.t. to their weights, i.e., the streams are encoded with the order $\pi^{{\star}}$ such that
    \begin{equation}\label{eq:WSROptimalEncodingOrder}
        \bm{H} = \begin{bmatrix}
            \bm{h}^{\T}_{\pi^{{\star}}(1)}\\
            \vdots\\
            \bm{h}^{\T}_{\pi^{{\star}}(2K)}
        \end{bmatrix} \text{with } \; \mu_{\pi^{{\star}}(1)} \le  \mu_{\pi^{{\star}}(2)} \le \dots \le \mu_{\pi^{{\star}}(2K)}. 
    \end{equation}
\end{lemma}
\begin{proof}
    To show this, we consider two encoding orders $\pi$ and $\sigma$ which only differ in the $i$-th and $i+1$-th stream being swapped and otherwise being identical, i.e.,
    \begin{equation}\label{eq:DecodingOrdersDefinition}
        \begin{aligned}
            \pi(m) &= \sigma(m)\quad\text{for}\quad {m \neq i, m\neq i+1}\\
            \pi(i) &= \sigma(i+1), \;\sigma(i)= \pi(i+1).
        \end{aligned}
    \end{equation}
    Maximizing \eqref{eq:WSRRateForOptimalScaling} is equal to minimizing
    \begin{equation}\label{eq:NegativeEquivalentWSRRateForOptimalScaling}
       \summe{k=1}{2K}\mu_{\pi(k)/\sigma(k)}\ln \left( s_{\pi/\sigma,kk}^2 \right)
    \end{equation}
    where $s_{\pi/\sigma,kk}$ are the diagonal entries of the QR decomposition $\bm{H}_{\pi/\sigma}^+= \bm{Q}_{\pi/\sigma} \bm{S}_{\pi/\sigma}$.
     To evaluate which encoding order results in a smaller value of \eqref{eq:NegativeEquivalentWSRRateForOptimalScaling}, we use the definition of the squared diagonal entries
     \begin{equation}\label{eq:DefinitionSquaredDiagonals}
        s^2_{\pi/\sigma,kk} = \bm{h}_{\pi/\sigma,k}^{\T}\left( \eye - \summe{n=1}{k-1}\bm{q}_{\pi/\sigma,n}\bm{q}_{\pi/\sigma,n}^{\T}\right) \bm{h}_{\pi/\sigma,k}^{\T}.
     \end{equation}
     Because of \eqref{eq:DefinitionSquaredDiagonals}, we directly see that
     \begin{equation}\label{eq:SquaredDiagonalEquivalence}
        s^2_{\pi,kk} = s^2_{\sigma,kk}\quad {m \neq i, m\neq i+1}.
     \end{equation}
    By using \eqref{eq:SquaredDiagonalEquivalence}, the difference in objective values of the two encoding orders is given by
     \begin{equation}\label{eq:WSRDifferenceOfEncodingOrders}
         \begin{aligned}
            \Delta_{\pi-\sigma}&=\summe{k=1}{2K}\mu_{\pi(k)}\ln \left( s_{\pi,kk}^2 \right) -  \summe{k=1}{2K}\mu_{\sigma(k)}\ln \left( s_{\sigma,kk}^2 \right)\\
            &= \summe{k=i}{i+1}\mu_{\pi(k)}\ln \left( s_{\pi,kk}^2 \right) -  \summe{k=i}{i+1}\mu_{\sigma(k)}\ln \left( s_{\sigma,kk}^2 \right)\\
            &= \mu_{\pi(i)}\ln \left( \frac{s_{\pi,ii}^2}{s_{\sigma,i+1,i+1}^2} \right) -  \mu_{\pi(i+1)}\ln \left( \frac{s_{\sigma,ii}^2}{s_{\pi,i+1,i+i}^2} \right)\\
         \end{aligned}
     \end{equation}
     where we used $\mu_{\pi(i)} = \mu_{\sigma(i+1)} $ and $\mu_{\pi(i+1)} = \mu_{\sigma(i)} $ due to \eqref{eq:DecodingOrdersDefinition} in the last line.
     Because $\det(\bm{H}_{\pi}^+\bm{H}_{\pi}^{+\T})=\det(\bm{H}_{\sigma}^+\bm{H}_{\sigma}^{+\T})$ and  because of  \eqref{eq:SquaredDiagonalEquivalence}, we have 
     \begin{equation}
        s_{\sigma,ii}^2 s_{\sigma,i+1,i+1}^2=s_{\pi,ii}^2s_{\pi,i+1,i+1}^2=c_i.
     \end{equation}
     Therefore, the difference in \eqref{eq:WSRDifferenceOfEncodingOrders} can be rewritten as
     \begin{equation}
        \Delta_{\pi-\sigma} = ( \mu_{\pi(i)} -  \mu_{\pi(i+1)})\ln \left( \frac{s_{\sigma,ii}^2 s_{\pi,ii}^2}{c_i} \right)\\
     \end{equation}
     According to \eqref{eq:DefinitionSquaredDiagonals}, we have $ {s_{\sigma,ii}^2 s_{\pi,ii}^2} \ge c_i$ and, therefore, $\ln \left( {s_{\sigma,ii}^2 s_{\pi,ii}^2}/{c_i} \right)\ge 0$.
  Hence, $\pi$ is preferred over $\sigma$ if
     \begin{equation}
        \Delta_{\pi-\sigma}  \le 0
     \end{equation}
     which means that 
     \begin{equation}
        \mu_{\pi(i)} \le  \mu_{\pi(i+1)}.
     \end{equation}
     This results directly in the optimal encoding order 
     \begin{equation}
         \pi^{\star}: \mu_{\pi^{\star}(1)} \le  \mu_{\pi^{\star}(2)} \le \dots, \le \mu_{\pi^{\star}(2K)}.
     \end{equation} 
\end{proof}
With the optimized scaling matrix and the optimized encoding order, we can generalize our results of the sum rate to the whole rate region according to the following proposition.
\begin{proposition}\label{prop:WSRLRNoImpact}
    For the whole rate region, \ac{LR}-aided \ac{NP} as well as all other{UTLR-invariant} methods lead to exactly the same solution as the conventional algorithm without \ac{LR} when considering the optimized scaling in Lemma \ref{lem:WSROptimalScaling} and the optimized encoding order in Lemma \ref{lem:OptimizedEncodingOrder}.
\end{proposition}
\begin{proof}
    The precoding matrix with the optimized scaling in Lemma \ref{lem:WSROptimalScaling} and the optimized encoding order in Lemma \ref{lem:OptimizedEncodingOrder} is given by 
    \begin{equation}
       \bm{P} =  \bm{Q}\bm{R} \bm{D}^{{\star}}\quad
\text{with} \quad
        d_{ii}^{{\star}} = \pm \sqrt{\frac{\mu_{\pi^{\star}(i)}}{r^2_{\pi^{\star},ii}} }.
    \end{equation}
    Because  $\mu_{{\pi^{\star}(i)} }\le \mu_{\pi^{\star}(i+1)}$, we have 
    \begin{equation}
       \bm{P} \in \mathbb{U}
    \end{equation}
    and, therefore, the implications of Theorem \ref{theo:GeneralGreedy} hold.
    Please note that all points on the rate boundary which cannot be achieved via the weighted sum rate are achieved by time sharing. 
\end{proof}
}

{
    Similar to Proposition \ref{prop:LocalOptScalingLLLNP}, we can give the following Proposition regarding local optimality.
    \begin{proposition}
        Given the optimized encoding order from Lemma \ref{lem:OptimizedEncodingOrder}, the optimized scaling in Lemma \ref{lem:WSROptimalScaling} 
        \begin{equation}
            \bm{D}^{\star} = \diag^{\inv}(\bm{r})\diag(\bm{\sqrt{\mu^{\star}}})
        \end{equation}
       is locally optimal for the \ac{LLL}-aided \ac{NP} algorithm.
    \end{proposition}
    \begin{proof}
    The proof is similar to Proposition \ref{prop:LocalOptScalingLLLNP} with the scaling given as
    \begin{equation}\label{eq:LemmaLocalOptNPScalingOptDefinitionWSR}
        \bm{d}^{\star} = [{\sqrt{\mu_1^{\star}}}/{r_{11}},{\sqrt{\mu_2^{\star}}}/{r_{22}},\dots,{\sqrt{\mu_{2K}^{\star}}}/{r_{2K,2K}}]^{\T}.
    \end{equation}
    With the same argumentation as in Proposition \ref{prop:LocalOptScalingLLLNP}, the integer matrix $\bm{T}$ is constant within $B_{\delta}(\bm{d}^{\star})$.
    Therefore, we arrive at
        \begin{equation}
            \expct{\norm{\bm{\bar{x}}}^2} = \frac{1}{12}\summe{j=1}{2K} \tilde{r}^2_{jj} = \frac{1}{12}\summe{j=1}{2K} {r}^2_{jj} d_j^2 \quad \forall \bm{d} \in B_{\delta}(\bm{d}^{\star})
     \end{equation}
     where we exploited that $\bm{\tilde{R}} = \bm{R}\bm{D}\bm{T}$ for an upper triangular $\bm{T}$. 
     This is the exact same power as in case of the conventional \ac{NP} algorithm.
     Following Lemma \ref{lem:WSROptimalScaling}, the optimal scaling within $B_{\delta}(\bm{d}^{\star})$ is given by \eqref{eq:LemmaLocalOptNPScalingOptDefinitionWSR} and, hence, the scaling of Lemma~\ref{lem:WSROptimalScaling} is locally optimal.
    \end{proof}
}
{
    \subsection{Generalization to Practical Considerations}\label{subsec:PracCons}
    The analysis above has been performed based on the mutual information in the high-\ac{SNR} regime under certain model assumptions.
    Hence, the question arises if the results above can be generalized to more practical considerations like moderate \ac{SNR}, finite constellations, or finite block-lengths.
    Here, it has to be mentioned that once the lattice has the structural property $\bm{B} \in \unitriag$, all{UTLR-invariant} methods lead to exactly the same integer vector regardless if \ac{LR} is applied or not.
    We have shown the equivalence w.r.t. the optimization variable.
    Therefore, regardless which metric is used like e.g., moderate \ac{SNR}, finite constellations, or finite block-lengths, once $\bm{B} \in \unitriag$ there will be no difference for{UTLR-invariant} methods with or without \ac{LR}.
    However, it has to be evaluated if the optimized scaling in \eqref{eq:NearestPlaneOptScaling} leading to $\bm{B} \in \unitriag$ in the first place is still reasonable for the more practical considerations.
    }

    {
        In Section~\ref{sec:Numerical}, we will see that the scaling in \eqref{eq:NearestPlaneOptScaling} leading to $\bm{B} \in \unitriag$ is valid also for practical scenarios.
        It is important to note that for moderate or low \ac{SNR} values, a stream selection has to be performed. 
        The scaling in (22) assumes all streams to be active and, hence, (22) is chosen for the effective channel matrix arising from a user selection.
        In comparison to the optimal choice (for moderate and low SNR values, otherwise (22) is optimal), no performance degradation was observed when using (22) (see Section \ref{sec:Numerical}).
    }

    {
        However, there is still the question why the existing literature has observed large gains at high-\ac{SNR} when using the identity scaling together with \ac{LLL} reduction w.r.t. the \ac{BER}.
        In our analysis, regarding the mutual information, the identity scaling cannot be recommended.
        Neither from a theoretical point of view nor does it provide a good performance as can be seen in Section \ref{sec:Numerical}.
        The reason why large gains were observed is that a fixed constellation was used for the users.
        If the constellation is way too small for a given SNR value, the identity scaling leads to a better performance than the scaling in \eqref{eq:NearestPlaneOptScaling}.
        However, when optimizing the constellation instead of using a fixed one, the scaling in \eqref{eq:NearestPlaneOptScaling} is clearly outperforming the identity scaling as suggested by the results based on the mutual information.
    }

    {
        \subsection{Further Generalizations (MSE-optimizing Precoders)}
        The results of this paper can be further generalized. 
        An important aspect is the generalization to more complex precoding structures.
        While a \ac{ZF} precoder with stream allocation gives a good performance also for low and moderate \ac{SNR} values, more complex precoding structures like \ac{MSE}-optimizing schemes can still lead to small improvements.
        Similar to this article, we experienced in our simulations that also for more complex precoding schemes, there is no impact of \ac{LR}.
        However, an analysis of these precoding schemes is not within the scope of this article and a thorough investigation of these more complex precoding structures will be conducted in future work.
    }

\begin{figure}[t!]
    \centering
    \hspace*{15pt}
    \vspace*{2pt}
    \includegraphics*{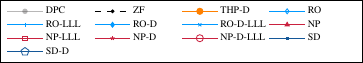}\\
    \includegraphics*{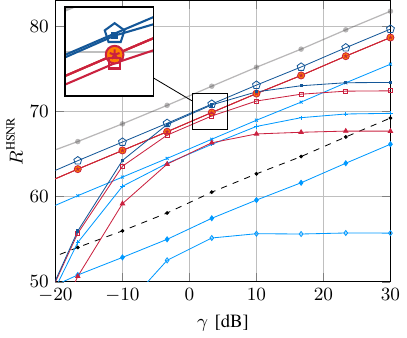}
    \caption{$\gamma=0$ dB, $\kappa=0$, $N=6$, $P_{\text{Tx}}=40\,\text{dB}$ }
    \label{fig:Gamma}
\end{figure}

\begin{figure}[b!]
    \centering
    \hspace*{15pt}
    \vspace*{2pt}
    \includegraphics*{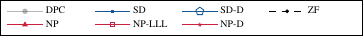}\\
    \includegraphics*{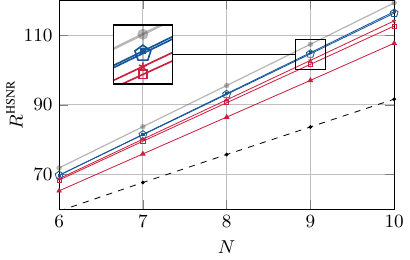}
    \caption{$\gamma = 0$ dB, $\kappa =0$ dB, $P_{\text{dB}} = 40$ dB}
    \label{fig:AntennaGammaZero}
\end{figure}
\section{Numerical Results}\label{sec:Numerical}
We now support the theoretical results of the last sections with simulations.
We consider a quadratic system where $\bm{H}_{\mathbb{C}} \in \cmplx{K \times K}$ {is the complex channel matrix.
A simple Rician model 
\begin{equation}
    \bm{H}_{\mathbb{C}} = \sqrt{\frac{\kappa}{1+\kappa}} \bm{1} \bm{1}^{\T} +  \sqrt{\frac{1}{1+\kappa}} \bm{H}^{\text{Ray}}_{\mathbb{C}} 
\end{equation}
is used where $\bm{H}_{\mathbb{C}}^{\text{Ray}} \in \cmplx{K \times K}$ has i.i.d. $\gaussdist{0}{1}$ entries and $\kappa$ is the Rician factor.
The Rician component is only introduced for controlling the conditioning of the channel matrix and, therefore, the all-ones vector is used for simplicity.
Hence, to worsen the conditioning of the channel matrix, we either choose a non-zero Rician factor $\kappa \neq 0$, or we use a Rayleigh channel model ($\kappa = 0$) but with imposing a path loss $\gamma$ to user $1$.}

Prior to applying the algorithms, a coding order optimization w.r.t. the \ac{MSE} is performed (see \cite{THPMichael}), which is similar to a V-BLAST ordering.
This is especially important for the moderate and low \ac{SNR} regions.
\subsection{Classical Algorithms}
For the algorithms, we use linear ZF and \ac{DPC}, both with Gaussian symbols, as a reference.
All other methods have uniformly distributed symbols and, hence, experience a shaping loss.
For the algorithms, we start with an identity scaling $\bm{D}$ and consider the RO and NP algorithms, their LLL-aided versions RO-LLL and NP-LLL, as well as the SD algorithm.
Additionally, we analyze the RO and NP methods with their optimal scaling matrix $\bm{D}$ from \eqref{eq:RoundingOffOptScaling} and \eqref{eq:NearestPlaneOptScaling} denoted as RO-D and NP-D, respectively.
Furthermore, we use the optimized scaling in \eqref{eq:NearestPlaneOptScaling} with the LLL-aided RO and NP algorithms together with the SD algorithm stated as RO-D-LLL, NP-D-LLL, and SD-D, respectively.
Lastly, we consider \ac{THP}-D with the optimal scaling in \eqref{eq:THPOptScaling}.
For the LLL algorithm, we use the implementation of \cite{WubbenLLLBetter}.

\begin{figure}[t!]
    \centering
    \hspace*{15pt}
    \vspace*{2pt}
    \includegraphics*{figures/LegendAnt.pdf}\\
    \includegraphics*{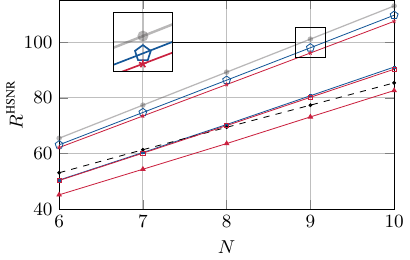}
    \caption{$\gamma = -20$ dB, $\kappa = 0$, $P_{\text{dB}} = 40$ dB}
    \label{fig:AntennaGammaTwenty}
\end{figure}

In Fig. \ref{fig:Gamma}, the high-\ac{SNR} asymptotes of all algorithms are compared with each other for the dimension $N=K=6$.
To see the importance of the matrix $\bm{D}$, the pathloss $\gamma$ of user 1 is modified.
We can observe that all algorithms including the scaling in \eqref{eq:NearestPlaneOptScaling} [or \eqref{eq:THPOptScaling} for THP-D] are robust w.r.t. a change of $\gamma$.
On the other hand, the algorithms relying on the identity scaling, experience a significant degradation in performance when $\gamma$ is modified.
Only for around $\gamma=0\,$dB, which corresponds to a well-conditioned channel, the identity scaling $\bm{D}=\eye$ is valid for this scenario.
For smaller values, the performance decreases rapidly, whereas for higher values, the rate saturates even though the channel gain increases.
The optimality of the scaling in \eqref{eq:NearestPlaneOptScaling} has been shown in Proposition \ref{prop:ScalingMatrixOptRate} in the case of the NP algorithm.
However, we can see that this is, in general, a good choice, as also the \ac{SD} algorithm benefits from it.
The SD-D algorithm clearly outperforms the SD algorithm, with the difference being especially pronounced if the channel is not well conditioned.
Beyond the advantages of the scaling in \eqref{eq:NearestPlaneOptScaling} [or \eqref{eq:THPOptScaling}], Fig. \ref{fig:Gamma} additionally allows to verify Theorem \ref{theo:LLLNPSameTHP} and Corollary \ref{col:LLLTHPSameTHP} by recognizing that THP-D, NP-D, and NP-D-LLL have the same performance.
As stated in Remark \ref{rem:OnlyForScaled}, this only holds for the algorithms relying on the scaling in \eqref{eq:NearestPlaneOptScaling} [or \eqref{eq:THPOptScaling}].
For the algorithms with identity scaling, we can see that \ac{LR} has a strong impact, and NP-LLL leads to a significant performance increase in comparison to NP.
However, as addressed in Remark \ref{rem:OnlyForScaled}, the performance of NP-LLL (and therefore also NP) is worse in comparison to NP-D-(LLL).
Furthermore, we can see that all RO algorithms lead to a worse performance than their NP counterpart.
\ac{LR} leads to a performance increase for all RO versions with RO-D-LLL, i.e., RO with the scaling in \eqref{eq:NearestPlaneOptScaling} together with \ac{LR}, achieving the best performance.
Nevertheless, verifying Proposition~\ref{prop:ROWorseNP}, RO-D-LLL cannot achieve a better performance than NP-D-(LLL) (or THP-D).

\begin{figure}[t!]
    \centering
    \hspace*{15pt}
    \vspace*{2pt}
    \includegraphics*{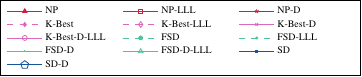}\\
    \includegraphics*{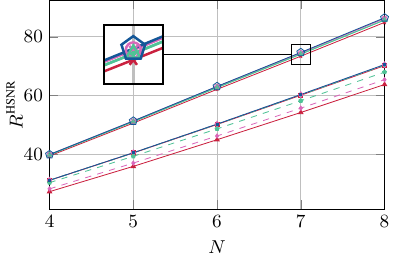}
    \caption{$\gamma = -20$ dB, $\kappa = 0$, $P_{\text{Tx}} = 40$ dB, $B=D=3$}
    \label{fig:C_AntennaMinus20}
    \vspace*{-0.35cm}
\end{figure}

\begin{figure}[b!]
    \centering
    \hspace*{15pt}
    \vspace*{2pt}
    \includegraphics*{figures/C_Legend.pdf}\\
    \includegraphics*{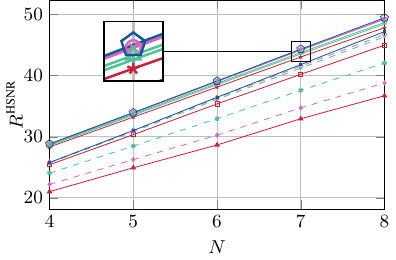}
    \caption{$\gamma = 0$ dB, $\kappa = 20$ dB, $P_{\text{Tx}} = 40$ dB, $B=D=3$}
    \label{fig:C_AntennaRician20}
\end{figure}

As NP-D-LLL, NP-D, and THP-D lead to exactly the same performance, we will only consider NP-D in the following figures.
Additionally, as the RO algorithms lead to a worse performance in comparison to their NP counterparts, the RO methods are also omitted in the following.

We will now further analyze the importance of the scaling matrix $\bm{D}$ and the influence of the channel condition.
According to Fig. \ref{fig:Gamma}, the channel condition has a huge impact on the choice of $\bm{D}$.
We have observed that the identity scaling leads to a significant degradation in the performance for all algorithms.
Only for well-conditioned channels, the identity matrix seems to be a good choice.
This is further illustrated in Figs. \ref{fig:AntennaGammaZero} and \ref{fig:AntennaGammaTwenty}, where we can see the high-\ac{SNR} asymptotes w.r.t. the number of base station antennas (equal to the number of users).
For the well-conditioned channels in Fig. \ref{fig:AntennaGammaZero} ($\gamma=0\,\text{dB}$), the scaling matrix $\bm{D}$ has only a slight impact on the performance.
However, when decreasing the channel gain of user 1 in Fig. \ref{fig:AntennaGammaTwenty} ($\gamma = -20\,\text{dB}$),
this is different, and the impact of optimizing the scaling matrix $\bm{D}$ is very pronounced.
This holds for all methods, including the \ac{SD} algorithm.
Actually, even the low-complexity method NP-D with the optimal scaling in \eqref{eq:NearestPlaneOptScaling} clearly outperforms the \ac{SD} with the identity scaling $\bm{D}=\eye$.

\begin{figure}[t!]
    \centering
    \hspace*{15pt}
    \vspace*{-3pt}
    \includegraphics*{figures/C_Legend.pdf}\\
    \subfigure[K-Best]{
    \includegraphics*{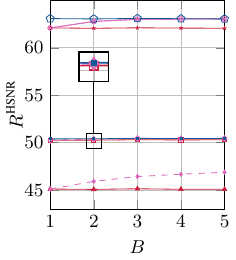}
    }\subfigure[FSD]{
        \includegraphics*{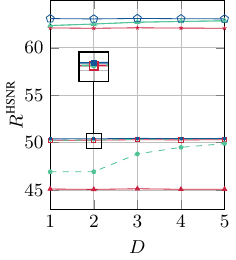}
    }
    \caption{$N=6$, $\gamma = -20$ dB, $\kappa = 0$ dB, $P_{\text{Tx}} = 40$ dB}
    \label{fig:HyperMinus20}
    \vspace*{-0.55cm}
\end{figure}

\subsection{Advanced Algorithms}
Having analyzed the classical {methods} like the \ac{RO} or \ac{NP} algorithm, we are now considering more advanced techniques.
In the following simulations, we are focusing on the K-Best and the \ac{FSD}, which we discussed in Section \ref{sec:Generalization}.
For the K-Best algorithm, we take the $B$ candidates in each step, where, if not stated otherwise, we set $B=3$.
Additionally, each candidate is expanded with $3$ children in each step.
Please see Section \ref{sec:Generalization} for details.
For the \ac{FSD}, we use $D$ dimensions, i.e., $a_{2K-D+1},\dots,a_{2K}$ for exhaustive search whereas the remaining dimensions, i.e., $a_{2K-D},\dots,a_{1}$, are given by the \ac{NP} algorithm for each of the possible candidates.
The exhaustive search is initialized with the NP-(D)-LLL solution, i.e., the constraint set $-\Delta \le a- a^{\text{NP-(D)-LLL}}_{i}\le \Delta$ is considered for $i=2K-D+1,\dots,{2K}$.
We choose $\Delta =1$ as we have observed that the solution is typically close to the NP-(D)-LLL solution.
This small value of $\Delta$ allows a higher number of $D$, i.e., the dimensions with exhaustive search.
We have experienced a better performance with a small $\Delta$ and higher $D$.
Please see Section \ref{sec:Generalization} for details.

The K-Best as well as the \ac{FSD} are analyzed with and without \ac{LR} and, additionally, with and without the scaling matrix in \eqref{eq:NearestPlaneOptScaling}.
This results in the algorithms, K-Best, K-Best-LLL, K-Best-D, K-Best-D-LLL for the K-Best methods and in the algorithms FSD, FSD-LLL, FSD-D, FSD-D-LLL for the \ac{FSD} methods.
To allow a better comparison, we keep the SD-(D) as well as the NP-(D)-(LLL) methods.
Please note again that the NP-D-LLL, NP-D, and the THP-D have the same performance and, hence, only NP-D is given in the following plots.

In Fig. \ref{fig:C_AntennaMinus20}, we analyze the same setup as in Fig. \ref{fig:AntennaGammaTwenty}, where one user has an additional pathloss of $20$ dB.
We can again observe that the algorithms with the scaling in \eqref{eq:NearestPlaneOptScaling} achieve a much better performance than the algorithms with an identity scaling.
This time, the focus is on the more advanced techniques, i.e., the K-Best and the \ac{FSD}.
In case of an identity scaling, the conventional K-Best and the \ac{FSD} are both worse than their LLL-aided versions, K-Best-LLL and FSD-LLL.
The LLL-aided versions have a similar performance, almost overlapping with the optimal \ac{SD}.
Without \ac{LR}, the performance is worse for both algorithms, with the \ac{FSD} achieving a better performance than the K-Best method.
For the methods with the scaling in \eqref{eq:NearestPlaneOptScaling}, the situation is different.
Here, there is no difference between the conventional and the LLL-aided versions.
For the K-Best algorithms, this has been shown in Section \ref{sec:Generalization}, whereas for the \ac{FSD}, this can only be observed for this specific scenario.
Additionally, the K-Best algorithms K-Best-D and K-Best-D-LLL lead to a slightly better performance in comparison to FSD-D and FSD-D-LLL.

We are now considering a different scenario.
In the previous simulation, the condition of the channel matrix was reduced by adding an extra path loss of $20$ dB for one user.
Now, the condition is reduced differently.
Instead of adding an extra path loss, the correlation of the users is {increased by choosing a non-zero Rician factor.}

\begin{figure}[t!]
    \centering
    \hspace*{15pt}
    \vspace*{-3pt}
    \includegraphics*{figures/C_Legend.pdf}\\
    \subfigure[K-Best]{
    \includegraphics*{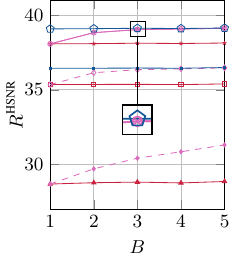}
    }\subfigure[FSD]{
        \includegraphics*{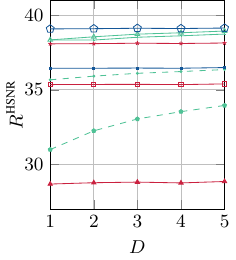}
    }
    \caption{$N=6$, $\gamma = 0$ dB, $\kappa = 20$ dB, $P_{\text{Tx}} = 40$ dB}
    \label{fig:HyperRician20}
\end{figure}

Fig. \ref{fig:C_AntennaRician20} shows the performance of the different methods for a Rician factor of $\kappa = 20$ dB.
We observe again that the methods with the identity scaling perform worse in comparison to the ones based on \eqref{eq:NearestPlaneOptScaling}.
However, the more advanced LLL-aided methods K-Best-LLL, \ac{FSD}-LLL, and especially the \ac{SD} algorithm improve for an increasing dimension of the channel matrix.
For the simple NP-LLL, this improvement cannot be observed.
As in the scenario before, \ac{LR} improves the algorithms considerably in the case of the identity scaling.
Interestingly, the K-Best leads again to a worse performance than FSD.
However, the LLL-aided version K-Best-LLL achieves a higher rate than FSD-LLL.
When focusing on the algorithms based on the scaling \eqref{eq:NearestPlaneOptScaling}, we see again that the K-Best-D achieves the same rate as K-Best-D-LLL according to Section \ref{sec:Generalization} and, additionally, leads to a better performance than the \ac{FSD} methods FSD-D and FSD-D-LLL.
In comparison to Fig. \ref{fig:C_AntennaMinus20}, a slight difference appears between FSD-D-LLL and FSD-D.

We analyze now the performance of the K-Best and \ac{FSD} method w.r.t. their complexity parameters $B$ and $D$ for both scenarios above and a dimension of $N=6$.
Starting with the i.i.d. Rayleigh scenario where one user has an additional pathloss of $20$ dB, we can see the dependence on the parameters $B$ and $D$ for K-Best and FSD in Fig. \ref{fig:HyperMinus20}.
We can see that the K-Best improves a lot faster than the \ac{FSD}.
It has to be noted that the computational complexity scales polynomially for K-Best with $B$ but exponentially for the \ac{FSD} with $D$.
Only the K-Best with the identity scaling and without \ac{LR} has a poor performance, significantly worse than the FSD method (also with identity scaling and without \ac{LR}).

The same behaviour can be seen in Fig. \ref{fig:HyperRician20} where all users have the same pathloss, but a Rician model with Rician factor $\kappa=20$ dB is considered.
Here, FSD-D and FSD-D-LLL lead to different rates.

{
    \subsection{Rate Region}
We are now validating Proposition \ref{prop:WSRLRNoImpact} where the results were generalized from the sum rate to the whole rate region.
In Fig. \ref{fig:RateRegionIID}, we see a two user scenario where their individual rates are given in a rate plot.
NP-D-SR is the NP algorithm for the sum rate (see Proposition \ref{prop:ScalingMatrixOptRate}) which was already condidered previously.
NP-D-WSR is the NP algorithm for the weighted sum rate (see Proposition \ref{prop:WSRLRNoImpact}).
Please note that also the optimized encoding order in Lemma \ref{lem:OptimizedEncodingOrder} is considered.
We can see that \ac{LR} gives no improvement for the whole rate region resulting in the exact same integer vector.
For the identity scaling, we already know that \ac{LR} yields an improvement, but both NP and NP-LLL do no not lie on the boundary, highlighting their suboptimality.
This suboptimality would be enhanced when considering a worse conditioned matrix as was already analyzed for the sum rate.
It has to be noted that in order to achieve the linear connection between the NP-D-SR sum rate points, time sharing has to be considered.
}
{
    \subsection{Practical Considerations}
    We demonstrate now that the results according to Theorem~\ref{theo:GeneralGreedy} are generalizeable to any metric and any practical scenario as noted in  Section \ref{subsec:PracCons}.
    The reason is that once $\bm{B} \in \unitriag$, the integer solution vector with \ac{LR} is exactly the same without \ac{LR} and the same performance arises regardless of the evaluation metric.
    Hence, the scaling in \eqref{eq:NearestPlaneOptScaling}, leading to $\bm{B} \in \unitriag$ in the first place, is now evaluated.
    In Fig. \ref{fig:DiagScaling}, we consider finite \ac{SNR}.
For finite \ac{SNR}, we use again the \ac{ZF} precoder but with an additional stream selection.
All possible allocations are considered and for each, the \ac{ZF} precoder is used. 
Then, the best one is taken. 
For each allocation, we use the identity scaling for the NP(-LLL) algorithm whereas the scaling in \eqref{eq:NearestPlaneOptScaling} is used for NP-D(-LLL).
In order to evaluate the scaling in \eqref{eq:NearestPlaneOptScaling}, we also consider NP-DOpt(-LLL).
Here, we use \eqref{eq:NearestPlaneOptScaling} as an initialization but then a local optimization is used in order to improve on NP-D(-LLL).
We can see that the scaling in \eqref{eq:NearestPlaneOptScaling} leading to $\bm{B} \in \unitriag$ is valid over the whole \ac{SNR} region clearly outperforming $\bm{D}=\eye$.}
\begin{figure}[t!]
    \centering
    \vspace*{2pt}
    \includegraphics*[scale =1]{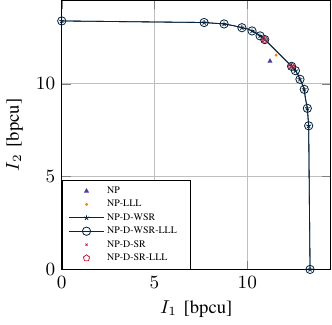}
    \caption{$N=K=2$, $\gamma = 0$ dB, $\kappa = 0$, $P_{\text{Tx}} = 40$ dB}
    \label{fig:RateRegionIID}
    \vspace*{-0.35cm}
\end{figure}

{
With a stream allocation also non-quadratic scenarios can be considered like $K=12$ and $N=6$ in Fig. \ref{fig:Overloaded}. 
However, due to the size of the system, we consider a stream selection similar to \cite{VPSumRate} where the users are successively added in a greedy manner such that the rate is maximized.
Even though we have more users than transmit antennas, the observations remain the same.
The optimized scaling \eqref{eq:NearestPlaneOptScaling} leads to a better performance than the identity scaling, with \ac{LLL} reduction providing no improvement.}

{
In the simulations so far, we assumed a uniform distribution over the Voronoi region.
However, in practice, a discrete constellation will be used.
In Fig. \ref{fig:FiniteConst}, we can see the performance for finite QAM constellations where NP-D and NP-D-LLL are again leading to the same performance regardless of the evaluation metric because $\bm{B} \in \unitriag$.
Additionally, the optimized scaling in \eqref{eq:NearestPlaneOptScaling} leading to $\bm{B} \in \unitriag$ is valid over the whole \ac{SNR} range, clearly outperforming the identity scaling.
However, this is only true under the assumption that, depending on the \ac{SNR}, the appropriate constellation is chosen and that the constellation is not too small for a given \ac{SNR} point.}
\begin{figure}[t!]
    \centering
    \hspace*{15pt}
    \vspace*{2pt}
    \includegraphics*{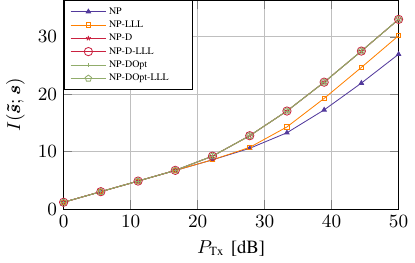}
    \caption{$\gamma = 0$ dB, $\kappa = 20$ dB, $N=K=3$}
    \label{fig:DiagScaling}
\end{figure}
\begin{figure}[b!]
    \centering
    \vspace*{2pt}
    \includegraphics*[scale =1]{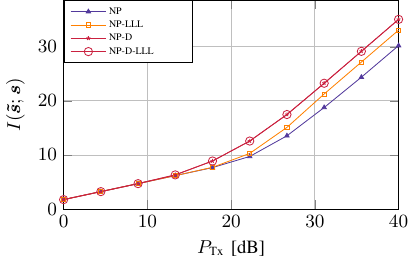}
    \caption{Overloaded system with $N=6$ and $K=12$ with $\gamma=0$ dB and $\kappa = 15$ dB and a ZF scheme with successive user selection.}
    \label{fig:Overloaded}
    \vspace*{-0.35cm}
\end{figure}
{
In the existing literature, large gains for the identity scaling with \ac{LLL} reduction has been shown w.r.t. the \ac{BER}.
However, fixed constellations were chosen in these simulations.
In practice, each user would adapt their constellation depending on their link quality.
Regarding Fig. \ref{fig:FiniteConst} for the 16-QAM constellation, the identity scaling with \ac{LR} achieving an equal effective \ac{SNR} per user channel, outperforms the rate optimal diagonal scaling according to \eqref{eq:NearestPlaneOptScaling} for a high enough \ac{SNR} value.
However, at such a high \ac{SNR}, a much higher constellation should have been chosen.
Hence, when adjusting the modulation depending on the \ac{SNR}, the optimized scaling \eqref{eq:NearestPlaneOptScaling} is clearly outperforming the identity scaling.
Actually, in comparison to Fig. \ref{fig:FiniteConst} where each user has the same constellation, each user should be allowed to have an individual constellation.
The reason is that in comparison to the identity scaling, the scaling \eqref{eq:NearestPlaneOptScaling} leads to different effective \acp{SNR} for the users.
Users with a good link quality need a large constellation and users with a poor link quality need a small constellation.
For the identity scaling, where each user channel has the same \ac{SNR}, also the same constellation can be chosen.}
\\
\begin{figure}[t!]
    \centering
    \vspace*{2pt}
    \includegraphics*[scale =1]{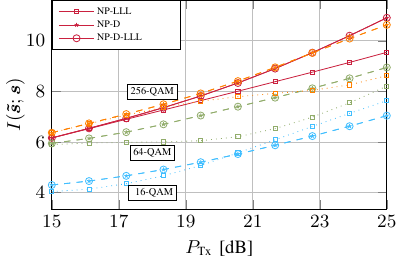}
    \caption{$\gamma = 0$ dB, $\kappa=20$ dB, $N=K=3$}
    \label{fig:FiniteConst}
    \vspace*{-0.35cm}
\end{figure}
{We have verified the results now for finite \ac{SNR} and finite constellations. 
However, we still used the mutual information as our performance metric.
Therefore, we evaluate our findings now with a practical scenario.
In particular, we illustrate that the theoretical results in this paper also hold in a practical scenario with finite \ac{SNR}, finite constellations, and a finite blocklength using channel coding.
For this, we consider a two-user scenario with 396 symbols per block (and user). 
    The \ac{BS} has $N=6$ antennas and can select between a 4-QAM, 16-QAM, 64-QAM, or 256-QAM symbol for each user individually according to their effective channel qualities.
    For the overall symbol including the precoder and lattice perturbation a power budget of $P_{\text{Tx}}$ is used.
    The noise is modeled as i.i.d. Gaussian with unit-variance (complex domain, 0.5 for each real axis).
    A LDPC $\frac{1}{3}$ code is used with the amount of bits given in the following table.
        \begin{table}[h!]
            \centering
                    \begin{tabular}{|r|r|r|}
            \hline
            constellation&coded-bits&info-bits\\
            \hline
            4-QAM&792&264\\
            \hline
            16-QAM&1584&528\\
            \hline
            64-QAM&2376&792\\
            \hline
            256-QAM&3168&1056\\
\hline
        \end{tabular}
        \caption{Number of bits per symbol.}
        \label{Table:BitTable}
    \end{table}
    \begin{figure}[b!]
        \centering
        \hspace*{15pt}
        \vspace*{-3pt}
        \includegraphics*{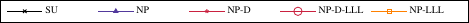}\\
        \subfigure[$\gamma=0$]{
        \includegraphics*{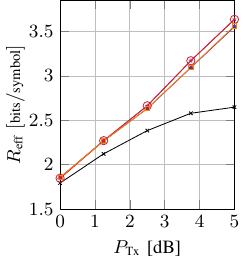}
        \label{fig:PracticalIID}
        }\subfigure[$\gamma = -10$ dB]{
            \includegraphics*{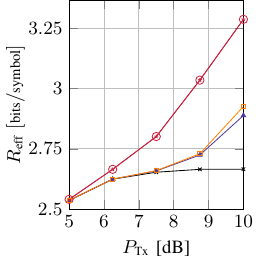}
            \label{fig:PracticalMinus10}
        }
        \caption{Practical LDPC-coded scenario for two users.}
        \label{fig:Practical}
    \end{figure} 
    To obtain the performance metric, we sum all of the successfully transmitted bits to both users.
    A block is only considered if it is perfectly decoded without any bit errors contributing either the total amount of information bits according to TABLE \ref{Table:BitTable} or 0 bits.
    Afterwards, the number of bits is divided by the number of symbols (396) giving the effective rate $R_{\text{eff}}$.
    After averaging over 1000 channel realizations, Fig. \ref{fig:PracticalIID} and Fig. \ref{fig:PracticalMinus10} are obtained.
    For each user, the optimal constellation is chosen.
    This is done by trying all constellations for each channel realization (16 possibilities for 2 users) and then taking the largest one for which the block can be decoded.
    Additionally, a user allocation is performed which allocates either the best single user or both users. 
    SU is also given which is the best single user rate.
}

{
    Firstly, we use i.i.d. Rayleigh fading in Fig.~\ref{fig:PracticalIID} and we can see that the observations are very similar to the theoretical considerations regarding the rate.
    Also in this LDPC-coded scenario, the optimized scaling \eqref{eq:NearestPlaneOptScaling} leads to a better performance than the identity scaling with LLL providing no improvement.
    }
{
    In Fig.~\ref{fig:PracticalMinus10}, we additionally give one of the two users a $10$ dB extra pathloss and by that worsening the channel condition.
    Similar to the simulations regarding the rate, we can see that also in this practical scenario, the identity scaling has a large performance degradation and is not suited if the channel has a worse conditioning.
}

\section{Conclusion}
{
We have demonstrated that for a certain lattice structure \ac{LR} has no effect on a complete class of algorithms.
This lattice structure appears naturally when optimizing the ZF precoder in \ac{VP} w.r.t. the mutual information.
In comparison to existing literature where \ac{SER} or \ac{BER} were considered, we showed that \ac{LLL} reduction is not able to improve the performance when the optimized ZF precoder is chosen.
With that, we show that already very simple algorithms like THP or NP are very effective.
Hence, for system designers targeting achievable rate with ZF precoding, this work suggests that the computational overhead of \ac{LR} can be entirely avoided by using simple, optimized \ac{THP}/\ac{NP}, provided the rate allocation matrix is correctly chosen.
The results of this work can be extended to more complex precoders.
In future work, we will focus on the extension to \ac{MSE}-based precoder design in the moderate-\ac{SNR} regime.
Additionally, the extension for \ac{VP} over multiple time instances will be considered in future work.
}

\vspace*{-0.1cm}

\bibliographystyle{IEEEtran}
\bibliography{refs}
\end{document}